\documentclass[journal,twoside,web]{ieeecolor}
\usepackage{generic}
\usepackage{cite}
\usepackage{amsmath,amssymb,amsfonts}
\usepackage{algorithm}
\usepackage{algorithmic}
\usepackage{graphicx}
\usepackage[caption=false,font=footnotesize]{subfig}
\usepackage{hyperref}
\hypersetup{
	colorlinks=true,
	citecolor=blue,
	linkcolor=blue,
	urlcolor=blue
}
\hypersetup{hidelinks}
\usepackage{textcomp}

\newcommand{\X}{\mathcal{X}}
\newcommand{\E}{\mathbb{E}}
\newcommand{\projX}{\Pi_{\X}}

\newtheorem{definition}{Definition}
\newtheorem{proposition}{Proposition}
\newtheorem{remark}{Remark}
\newtheorem{assumption}{Assumption}
\newtheorem{lemma}{Lemma}
\newtheorem{theorem}{Theorem}
\newtheorem{corollary}{Corollary}

\begin{document}

\title{Federated Zeroth-Order Optimization with Direction Aggregation and Variance Reduction}
	\author{Qianlong Dang,~\IEEEmembership{Member,~IEEE},~Baosheng Li
	\thanks{This work was supported in part by the National Nature Science Foundation of China under Grant 12501711 and 126011000, in part by the Postdoctoral Research Project of Shaanxi Province Grant 2025BSHEDZZ012, in part by the China Postdoctoral Science Foundation under Grant 2024M762632, in part by the Social Science Foundation of Shaanxi under Grant 2024R027, in part by Chinese Universities Scientific Fund Grant 2452026073, and in part by the Fundamental Research Funds for the Central Universities Grant 2452026134.
	}
	\thanks{
		Qianlong Dang and Baosheng Li are with the College of Science, Northwest A\&F University, Yangling 712100, China (e-mail: xidianqldang@163.com, baoshengli96@nwafu.edu.cn).}
}

\maketitle

\begin{abstract}
We study constrained nonsmooth nonconvex stochastic optimization in federated settings, where clients access only stochastic function evaluations. Existing federated zeroth-order methods primarily combine local zeroth-order updates with model averaging. However, they struggle with client drift induced by local projected updates and sampling variance in stochastic zeroth-order estimates. In this paper, we propose a federated zeroth-order framework based on direction aggregation, equipped with two local sampling schemes. Specifically, FedZOO constructs a local direction from a minibatch shared across multiple spherical queries, thereby controlling the stochastic errors arising from data sampling and gradient approximation. In contrast, FedVRZO forms its local estimator from independent sample--direction pairs, so that its sampling error is controlled directly by the number of sample--direction pairs. In both algorithms, clients compute the mean of the directions evaluated along their projected local trajectories, and the server performs a single projected update using the weighted aggregate of the client directions. Furthermore, we establish convergence guarantees for FedZOO and FedVRZO, respectively. Experiments on black-box adversarial attacks demonstrate the effectiveness of the proposed methods.\end{abstract}


\section{Introduction}
\label{sec:introduction}
\IEEEPARstart{F}{ederated} optimization \cite{mcmahan2017fedavg,reddi2021adaptive,wu2023fasterfed} seeks a shared decision vector by coordinating client-specific objectives while preserving data privacy through local computation \cite{gratton2021privacyzo}. In this paper, we consider the following finite-sum stochastic optimization problem
\begin{equation}
\begin{aligned}
    \min_{x\in\X}\quad
    f(x) &:= \sum_{i=1}^{M}p_i f_i(x), \\
    f_i(x) &:= \E_{\xi_i\sim\mathcal{D}_i}
    [F_i(x,\xi_i)],
\end{aligned}
\label{eq:intro_problem}
\end{equation}
where \(M\) is the number of clients, \(p_i>0\) with \(\sum_{i=1}^{M}p_i=1\), and \(\X\subseteq\mathbb{R}^{n}\) is a nonempty closed convex constraint set shared by all clients. The local objective \(f_i\) is defined as the expectation of the sample loss \(F_i(x,\xi_i)\) over the random variable \(\xi_i \sim \mathcal D_i\), and \(F_i\) may be nonsmooth and nonconvex. Each client \(i\) can only query stochastic function evaluations \(F_i(x,\xi_i)\), whereas first-order gradient information is unavailable. Thus, standard gradient-based federated algorithms become inapplicable, and stationarity can no longer be characterized by the classical gradient norm. In addition, the client data distributions \(\mathcal D_i\) may not be identical \cite{li2020fedprox,10508084}. Throughout the paper, \(\|\cdot\|\) denotes the Euclidean norm unless another norm is indicated.

\subsection{Related Work}
The absence of gradient information motivates centralized zeroth-order
optimization (ZOO) \cite{huang2024zoadmm}, where true gradients are replaced by finite differences of
function values along sampled directions. To handle nonsmoothness \cite{kungurtsev2021weaklyconvexzo}, randomized smoothing constructs a differentiable surrogate of the stochastic objective, allowing two-point function-value differences \cite{ghadimi2016mini} to approximate the gradient of the smoothed function. 
Ghadimi and Lan
\cite{ghadimi2013zo} developed a randomized zeroth-order algorithms for nonconvex stochastic programming, establishing stationarity complexity bounds with explicit dimension dependence.
Nesterov and Spokoiny \cite{nesterov_random_2017} analyzed Gaussian smoothing for gradient-free minimization and quantified the effects of the smoothing radius and directional sampling.
In nonsmooth nonconvex settings, gradient-free algorithms replace classical gradient-norm criteria with smoothed or neighborhood-based measures of stationarity \cite{lin2022gradientfree,chen2023fasterzo,kornowski2024zero}.
To enforce feasibility \cite{nguyen2022functionalzo,li2022functionalzo}, projected ZOO methods control stationarity over a closed convex constraint set \cite{balasubramanian2022zo,marrinan2026zo,chen2025projectionzo,maass2021symmetriczo}, whereas zeroth-order Frank--Wolfe variants offer a projection-free alternative \cite{akhtar2022fw}.
Centralized ZOO methods share a unified single-trajectory recursion, where local directions are estimated at the current iterate and projected once per iteration. This unified recursion allows stationarity to be naturally characterized through the projected gradient mapping of the smoothed objective. However, the same property disappears when clients perform multiple local updates between server synchronizations.

\begin{table*}[ht]
	\caption{Comparison of representative federated zeroth-order methods}
	\label{tab:closest_work}
	\centering
	\renewcommand{\arraystretch}{1.15}
	\setlength{\tabcolsep}{16pt}
	\begin{tabular}{ccccccc}
		\hline
		Approach & \shortstack{Nonsmooth\\objective}
		& \shortstack{Constrained\\problem}
		& \shortstack{Multiple\\local steps}
		& \shortstack{Direction\\aggregation}
		& \shortstack{Global\\projection}
		& \shortstack{Sample--direction\\pairs} \\
		\hline
		FedZO \cite{fedzo2022}
		& \(\times\) & \(\times\) & \(\checkmark\)
		& \(\times\) & \(\times\) & \(\times\) \\
		Qiu et al. \cite{qiu2023fedrzo}
		& \(\checkmark\) & \(\checkmark\) & \(\checkmark\)
		& \(\times\) & \(\times\) & \(\times\) \\
		Chen et al. \cite{chen2023finegrained}
		& \(\times\) & \(\times\) & \(\times\)
		& \(\checkmark\) & \(\times\) & \(\times\) \\
		FedZOO (Our)
		& \(\checkmark\) & \(\checkmark\) & \(\checkmark\)
		& \(\checkmark\) & \(\checkmark\) & \(\times\) \\
		FedVRZO (Our)
		& \(\checkmark\) & \(\checkmark\) & \(\checkmark\)
		& \(\checkmark\) & \(\checkmark\) & \(\checkmark\) \\
		\hline
	\end{tabular}
\end{table*}

The single-trajectory property fails when function evaluations are carried out by multiple clients rather than at a common iterate.
Early works \cite{yi2022automatica,yuan2023automatica} establish consensus-based architectures \cite{8629972} and primal--dual frameworks \cite{9535301} for distributed nonconvex zeroth-order optimization. Recent extensions \cite{10645219,he2024parallelzo,hu2023tsgzo,dang2024iotzo} have integrated communication compression \cite{11165038,xu2024qzogt} and variance-reduction techniques \cite{11086507,11601068} to reduce consensus and sampling errors. In contrast, federated zeroth-order methods replace continuous coordination with periodic server synchronizations. In every round, each client evaluates several local zeroth-order directions \cite{10487872}, so that the query points may drift from the server iterate and form client-specific projected trajectories. This divergence is affected by nonidentical distributions \(\mathcal{D}_i\) and client-specific stochastic sampling \cite{stich2019local,li2020fedprox,karimireddy2020scaffold,10508084}. Consequently, the uploaded direction may differ from the global smoothed gradient because of local trajectory displacement, sampling error, and heterogeneity across client objectives.

Existing federated zeroth-order methods address these issues under restricted or unconstrained settings. A comparison of representative federated zeroth-order methods is summarized in Table~\ref{tab:closest_work}.
For instance, Fang et al. \cite{fedzo2022} developed the FedZO algorithm to allow multiple local steps, but their reliance on terminal model differences confounds local directions with cumulative trajectory drift. Qiu et al. \cite{qiu2023fedrzo} study nondifferentiable nonconvex and hierarchical federated optimization through randomized smoothing and Clarke stationarity, but their framework relies on hierarchical objectives and does not target direction aggregation under a shared closed convex constraint. Building on this line of work, Chen et al. \cite{chen2023finegrained} analyze FedZO generalization and stability under smoothness and heavy-tailed noise assumptions. However, these methods fail to account for the scenario where local directions evaluated along client-specific projected trajectories are aggregated prior to a server-side projection. Consequently, it remains a critical theoretical challenge to establish convergence guarantees for this direction-aggregation paradigm.

\subsection{Motivation}

Federated zeroth-order methods reduce communication by performing multiple local
function-value queries, but methods with multiple local steps commonly
aggregate terminal client models
\cite{fedzo2022,qiu2023fedrzo}. Under constraints, each terminal model mixes
the evaluated directions with the displacement and projection effects
accumulated along a client-specific trajectory, which are further influenced
by client drift
\cite{stich2019local,li2020fedprox,karimireddy2020scaffold}. Since projection
and weighted averaging do not commute, the averaged model generally cannot be
interpreted as a projected step along the weighted global direction. This
precludes a direct application of centralized projected ZO analysis
\cite{balasubramanian2022zo,marrinan2026zo} and motivates aggregating local
directions before performing the server projection.

\subsection{Main Contributions}
This paper concerns constrained nonsmooth nonconvex stochastic optimization in
a synchronous federated setting, where each client accesses only stochastic function evaluations. 
In contrast to FedZO that aggregates local models, we consider a novel recursion in which
clients aggregate zeroth-order estimates evaluated along their projected local
trajectories and the server performs one weighted projected update. This architecture keeps
the server recursion on the shared feasible set while exposing the effects of
local drift and stochastic sampling. The main characteristics of the proposed methods and the corresponding analysis are summarized as follows.
\begin{enumerate}
	\item We propose a federated zeroth-order optimization (FedZOO) algorithm based on local-direction
	aggregation. Each client uses a common minibatch of
	\(B_t\) samples together with \(Q_t\) spherical directions to form a two-point estimate and then uploads the mean of the \(K\) directions evaluated along its projected local trajectory. Then, the server aggregates these directions with the objective weights and performs one projection onto the shared constraint set. 
	
	\item We propose a federated variance-reduced zeroth-order (FedVRZO) algorithm, which retains the same client-server update recursion but constructs each local estimator from \(N_t\) independent sample--direction pairs. The samples are drawn from \(\mathcal D_i\), and the directions are drawn uniformly from the unit sphere. This construction produces a conditional mean-zero sampling error whose second moment is of order \(\mathcal{O}(n/N_t)\).
	
	\item We establish the finite-time convergence guarantees for FedZOO and FedVRZO under nonsmooth nonconvex settings. For both algorithms, our direction-error analysis quantifies the coupled effects of projected local-trajectory drift, smoothed-gradient heterogeneity, and sampling-scheme variance.
	
	\item We characterize the communication and function-query complexities of
	FedZOO and FedVRZO. Under a fixed local budget of \(P\) two-point
    differences, the sampling-error factors of FedZOO and FedVRZO scale as
    \(P^{-1/2}\) and \(P^{-1}\), respectively. The resulting function-query complexities are
\(\mathcal O(MKn\epsilon^{-2}+Mn\varkappa^2\epsilon^{-6}/K)\) for FedZOO and
\(\mathcal O(MKn\epsilon^{-2}+Mn\varkappa\epsilon^{-4})\) for FedVRZO.
	We validate the performance of FedZOO and FedVRZO through comparative experiments with existing federated zeroth-order baselines on black-box adversarial attacks.
\end{enumerate}

\subsection{Organization}
Section~\ref{sec:preliminaries} introduces the smoothing construction, ZO
estimators, and stationarity measure. Section~\ref{sec:algorithms} gives the client and
server recursions, and Section~\ref{sec:convergence} establishes their
convergence and complexity bounds. Section~\ref{sec:experiments} describes the
numerical evaluation, and Section~\ref{sec:conclusion} concludes the paper.
Complete proofs are collected in the appendix.

\section{Preliminaries}
\label{sec:preliminaries}

\subsection{Spherical Smoothing and Two-Point Estimation}
In federated optimization, local loss functions are frequently nonsmooth and nonconvex, while first-order gradient information is unavailable. In this context, conventional finite-difference schemes fail to construct stable gradient approximations. Randomized smoothing addresses this issue by replacing the nonsmooth function with a differentiable surrogate, whose gradient can be estimated via local function evaluations \cite{ghadimi2013zo,nesterov_random_2017,marrinan2026zo}. To formalize this construction, we define the unit ball and unit sphere in \(\mathbb R^n\) as
\[
\mathbb B:=\{u\in\mathbb R^n:\|u\|\leq 1\},\qquad
\mathbb S:=\{v\in\mathbb R^n:\|v\|=1\}.
\]
Let \(\mathcal U\subseteq\mathbb R^n\) be a common neighborhood of \(\X\) on
which the stochastic function values are available. We select a smoothing radius \(\eta_0>0\) such that \(\X_{\eta_0}:=\X+\eta_0\mathbb B\subseteq\mathcal U\). Hence, for any \(x\in\X\) and \(0<\eta<\eta_0\), the perturbed query points \(x\pm\eta v\) belong to the function-evaluation domain \(\mathcal U\). Let \(u\sim\operatorname{Unif}(\mathbb B)\) be uniformly distributed on the unit ball. The ball-smoothed local objective for client \(i\) is defined as
\begin{equation}
	f_{i,\eta}(x):=\E_{u}\bigl[f_i(x+\eta u)\bigr].
	\label{eq:local_smoothing_def}
\end{equation}
Correspondingly, the ball-smoothed global objective function is given by
\begin{equation}
	f_\eta(x):=\sum_{i=1}^{M}p_i f_{i,\eta}(x)=\E_u\bigl[f(x+\eta u)\bigr].
	\label{eq:global_smoothing_def}
\end{equation}
Under the local Lipschitz continuity and integrability assumptions detailed in Section~\ref{sec:convergence}, the smoothed function \(f_{i,\eta}\) is continuously differentiable. Its gradient admits an exact representation:
\begin{equation}
	\nabla f_{i,\eta}(x)
	=\frac{n}{2\eta}\E_{\xi_i,v}
	\left[
	\bigl(F_i(x+\eta v,\xi_i)-F_i(x-\eta v,\xi_i)\bigr)v
	\right],
	\label{eq:smoothing_gradient_identity}
\end{equation}
where \(v\sim\operatorname{Unif}(\mathbb S)\) is uniformly distributed on the unit sphere and independent of the random sample \(\xi_i\). By linearity, the smoothed global gradient satisfies \(\nabla f_\eta(x)=\sum_{i=1}^{M}p_i\nabla f_{i,\eta}(x)\).

The representation in \eqref{eq:smoothing_gradient_identity} motivates a stochastic finite-difference estimator using a single data sample \(\xi_i\sim\mathcal D_i\) and a single random direction \(v\sim\operatorname{Unif}(\mathbb S)\). Specifically, the elementary symmetric two-point stochastic gradient estimator is constructed as
\begin{equation}
	\widehat g_{i,\eta}(x;\xi_i,v)
	:=
	\frac{n}{2\eta}
	\bigl[F_i(x+\eta v,\xi_i)-F_i(x-\eta v,\xi_i)\bigr]v,
	\label{eq:elementary_zo_estimator}
\end{equation}
where the same sample \(\xi_i\) is evaluated at both perturbed points \(x+\eta v\) and \(x-\eta v\). Taking the expectation with respect to the independent sample--direction pair \((\xi_i,v)\) recovers the gradient of the smoothed local objective:
\begin{equation}
	\E_{\xi_i,v}
	[\widehat g_{i,\eta}(x;\xi_i,v)]
	=\nabla f_{i,\eta}(x).
	\label{eq:elementary_zo_expectation}
\end{equation}

\subsection{Clarke Stationarity and Projected Gradient Mapping}

For a locally Lipschitz but possibly nondifferentiable objective, stationarity
cannot be expressed through a classical gradient norm. We first recall the
neighborhood-based notion used for nonsmooth nonconvex optimization
\cite{clarke1990optimization,goldstein1977optimization}.

\begin{definition}[Clarke and \(\eta\)-Clarke stationarity]
\label{def:clarke_stationarity}
\normalfont
For a locally Lipschitz function \(f:\mathcal U\to\mathbb R\), its Clarke
generalized directional derivative at \(x\in\mathcal U\) along
\(d\in\mathbb R^n\) is
\begin{equation}
    f^{\circ}(x;d)
    :=\limsup_{\substack{y\to x\\ s\downarrow 0}}
    \frac{f(y+s d)-f(y)}{s}.
\label{eq:clarke_directional_derivative}
\end{equation}
The Clarke subdifferential and its \(\eta\)-neighborhood enlargement are given by
\begin{subequations}
\begin{align}
    \partial f(x)
    &:=\left\{\ell\in\mathbb R^n:
    f^{\circ}(x;d)\geq \ell^{\top}d,\ \forall d\in\mathbb R^n\right\},
    \label{eq:clarke_subdifferential}\\
    \partial_\eta f(x)&:=\operatorname{conv}\!\left(\bigcup_{\substack{y\in\mathcal U\\ \|y-x\|\leq\eta}}\partial f(y)\right).
    \label{eq:eta_clarke_subdifferential}
\end{align}
\end{subequations}
A point \(x\in\X\) is \(\eta\)-Clarke stationary for
\eqref{eq:intro_problem} if
\begin{equation}
    0\in\partial_\eta f(x)+\mathcal N_{\X}(x),
\label{eq:eta_clarke_stationarity}
\end{equation}
where the normal cone of \(\X\) at \(x\) is
\[
    \mathcal N_{\X}(x)
    :=\{w\in\mathbb R^n:\langle w,z-x\rangle\leq0,
    \ \forall z\in\X\}.
\]
\end{definition}

Condition~\eqref{eq:eta_clarke_stationarity} is the constrained counterpart of
approximate Clarke stationarity. The enlargement \(\partial_\eta f(x)\)
collects generalized gradients at points within distance \(\eta\) of \(x\),
whereas \(\mathcal N_{\X}(x)\) accounts for the geometry of the feasible set.
If \(x\) lies in the interior of \(\X\), then
\(\mathcal N_{\X}(x)=\{0\}\), and the condition reduces to
\(0\in\partial_\eta f(x)\). At a boundary point, a nonzero generalized
gradient may instead be balanced by an outward normal. Moreover,
\(\partial f(x)\subseteq\partial_\eta f(x)\) and every Clarke stationary
point is \(\eta\)-Clarke stationary, but the converse need not hold for a
fixed \(\eta>0\). Thus, \(\eta\) specifies a spatial relaxation of the
generalized gradient.

Spherical smoothing replaces the
nonsmooth objective by the continuously differentiable function \(f_\eta\),
whose gradient can be estimated by the two-point construction above. This
motivates measuring constrained stationarity of the smoothed problem through
the projected gradient mapping
\cite{balasubramanian2022zo,marrinan2026zo}. For any \(0<\eta<\eta_0\),
stepsize \(a>0\), and \(x\in\X\), the projected gradient mapping is given by
\begin{equation}
    \mathcal G_{\eta,a}(x)
    :=\frac{1}{a}\left[
    x-\projX\bigl(x-a\nabla f_\eta(x)\bigr)
    \right].
\label{eq:projected_gradient_mapping}
\end{equation}
The point inside the projection is an unconstrained gradient step for
\(f_\eta\), and the projection returns it to \(\X\). Consequently,
\(\mathcal G_{\eta,a}(x)\) is the rescaled discrepancy between \(x\) and this
projected step. If \(\X=\mathbb R^n\), the projection is the identity and
\(\mathcal G_{\eta,a}(x)=\nabla f_\eta(x)\). For a general closed convex set,
the mapping reflects the active constraint geometry: it may vanish at a
boundary point even when \(\nabla f_\eta(x)\neq0\). Although its magnitude
depends on \(a\), its zero set is independent of \(a>0\).

\begin{proposition}[Smoothing and Clarke stationarity]
\label{prop:stationarity_transfer}
Let \(f\) be locally Lipschitz on \(\mathcal U\). For every \(x\in\X\),
\(0<\eta<\eta_0\), and \(a>0\), the gradient of the ball-smoothed
objective satisfies
\begin{equation}
    \nabla f_\eta(x)\in\partial_\eta f(x),
\label{eq:smoothing_clarke_inclusion}
\end{equation}
and the projected gradient mapping obeys
\begin{equation}
\mathcal G_{\eta,a}(x)=0
\quad\Longleftrightarrow\quad
0\in\nabla f_\eta(x)+\mathcal N_{\X}(x).
\label{eq:mapping_stationarity_equivalence}
\end{equation}
Consequently, \(\mathcal G_{\eta,a}(x)=0\) implies that \(x\) satisfies
\eqref{eq:eta_clarke_stationarity}.
\end{proposition}

The inclusion in \eqref{eq:smoothing_clarke_inclusion} follows from the
spherical smoothing result in \cite{marrinan2026zo} and
\eqref{eq:mapping_stationarity_equivalence} is the projection optimality
condition for a closed convex set.

\begin{remark}[Role of \(\eta\)]
\label{rem:eta_stationarity}
\normalfont
The parameter \(\eta\) in Definition~\ref{def:clarke_stationarity} specifies
the neighborhood used to enlarge the Clarke subdifferential. For a prescribed
horizon \(T\), let \(R\) be independent of the algorithmic randomness and
uniformly distributed on \(\{0,\ldots,T-1\}\). Our finite-time criterion is
\begin{align}
    \E[\|\mathcal G_{\eta,a}(x_R)\|^2],
\end{align}
which measures projected stationarity of the smoothed constrained problem at
the randomized output.
\end{remark}

\section{Federated Zeroth-Order Algorithms}
\label{sec:algorithms}

This section first defines the direction-aggregation recursion shared by
FedZOO and FedVRZO. It then presents the estimator used by each method.
FedZOO reuses one data minibatch across multiple spherical directions.
FedVRZO draws independent sample--direction pairs. Both methods use the same
projected client trajectory and server update.
The parameter sequences satisfy \(0<\eta_t<\eta_0\),
\(\gamma_t,\alpha_t>0\), and positive integer sampling budgets.

\subsection{Direction-Aggregation Strategy}

At communication round \(t\), the server broadcasts \(x_t\in\X\) to all
clients. Let \(K\geq1\) be the number of local direction evaluations and let
\(\gamma_t>0\) be the local stepsize. Client \(i\) initializes
\(x_{i,t,0}=x_t\), evaluates a stochastic zeroth-order direction \(g_{i,t,k}\) at each
local query point, and constructs
\begin{subequations}
\begin{align}
x_{i,t,k+1}
&=\projX\bigl(x_{i,t,k}-\gamma_tg_{i,t,k}\bigr),
\quad 0\leq k<K-1,
\label{eq:local_projected_update}\\
d_{i,t}
&=\frac{1}{K}\sum_{k=0}^{K-1}g_{i,t,k}.
\label{eq:client_average_direction}
\end{align}
\end{subequations}
The vector \(d_{i,t}\) is the arithmetic mean of the \(K\) ZO directions
evaluated along client \(i\)'s local trajectory and collects them in one
\(n\)-dimensional message. For \(K>1\), the component directions are evaluated
at different local points; after the sampling error is separated, their mean
may differ from \(\nabla f_{i,\eta_t}(x_t)\), which produces local drift.

The server aggregates these vectors using the objective weights
\(\{p_i\}_{i=1}^{M}\):
\begin{equation}
    d_t=\sum_{i=1}^{M}p_i d_{i,t}
\label{eq:server_aggregated_direction}
\end{equation}
It then updates the global iterate according to
\begin{equation}
    x_{t+1}=\projX\bigl(x_t-\alpha_t d_t\bigr).
\label{eq:federated_server_update}
\end{equation}
The local stepsize \(\gamma_t\) sets the spacing of the client query points.
The server stepsize \(\alpha_t\) scales the aggregated direction. The clients
collect directions along their projected trajectories, and the server applies
one projection to their weighted aggregate. 

\subsection{FedZOO}

FedZOO constructs each local direction from a data minibatch and several
spherical directions. At query point \(x_{i,t,k}\), client \(i\) draws
\begin{equation}
\mathcal B_{i,t,k}:=\{\xi_{i,t,k}^{b}\}_{b=1}^{B_t},\qquad
\mathcal V_{i,t,k}:=\{v_{i,t,k}^{j}\}_{j=1}^{Q_t},
\label{eq:fedzoo_local_batches}
\end{equation}
where the \(B_t\) samples follow \(\mathcal D_i\). The \(Q_t\) directions are
independent and uniform on \(\mathbb S\). The minibatch defines the
sample-average function
\begin{equation}
    \widehat F_{i,t,k}(x)
    =\frac{1}{B_t}\sum_{b=1}^{B_t}F_i(x,\xi_{i,t,k}^{b}).
\label{eq:fedzoo_batch_function}
\end{equation}
Using this same minibatch for every direction, the gradient estimator is given by
\begin{equation}
\begin{aligned}
g_{i,t,k}=&\frac{n}{2\eta_tQ_t}\sum_{j=1}^{Q_t}
\bigl[\widehat F_{i,t,k}(x_{i,t,k}+\eta_tv_{i,t,k}^{j})\\
&\hspace{23mm}
-\widehat F_{i,t,k}(x_{i,t,k}-\eta_tv_{i,t,k}^{j})\bigr]
v_{i,t,k}^{j}.
\end{aligned}
\label{eq:fedzoo_local_direction}
\end{equation}
The two evaluations associated with one direction use the same data samples,
and the minibatch is reused across all \(Q_t\) directions. Taking expectation
over the directions first gives the gradient of the ball-smoothed version of
\(\widehat F_{i,t,k}\). Meanwhile, taking expectation over the minibatch gives
\(\E[g_{i,t,k}\mid\mathcal F_{t,k}]
=\nabla f_{i,\eta_t}(x_{i,t,k})\). The two sampling stages produce a
data-sampling error controlled by \(B_t\) and a direction-sampling error
controlled by \(Q_t\). Thus, one local direction requires \(2B_tQ_t\)
sample-function evaluations.

\begin{algorithm}[t]
\caption{FedZOO}
\label{alg:fedzoo}
\begin{algorithmic}[1]
\REQUIRE \(x_0\in\X\), rounds \(T\), local evaluations \(K\), weights
\(\{p_i\}_{i=1}^{M}\), and sequences
\(\{\eta_t,\gamma_t,\alpha_t,B_t,Q_t\}\)
\FOR{\(t=0,\ldots,T-1\)}
    \STATE Server broadcasts \(x_t,\eta_t,\gamma_t,B_t,Q_t\).
    \FOR{each client \(i=1,\ldots,M\) in parallel}
        \STATE Set \(x_{i,t,0}=x_t\) and \(d_{i,t}=0\).
        \FOR{\(k=0,\ldots,K-1\)}
            \STATE Draw \(\mathcal B_{i,t,k}\) and
            \(\mathcal V_{i,t,k}\).
            \STATE Compute \(g_{i,t,k}\) from
            \eqref{eq:fedzoo_local_direction} and set
            \(d_{i,t}\leftarrow d_{i,t}+g_{i,t,k}/K\).
            \IF{\(k<K-1\)}
                \STATE Set \(x_{i,t,k+1}\leftarrow
                \projX(x_{i,t,k}-\gamma_tg_{i,t,k})\).
            \ENDIF
        \ENDFOR
        \STATE Upload \(d_{i,t}\) to the server.
    \ENDFOR
    \STATE Server sets \(d_t\leftarrow\sum_{i=1}^{M}p_id_{i,t}\) and
    \(x_{t+1}\leftarrow\projX(x_t-\alpha_td_t)\).
\ENDFOR
    \STATE Draw \(R\) uniformly from \(\{0,\ldots,T-1\}\), independently of
    all algorithmic randomness.
\ENSURE \(x_R\)
\end{algorithmic}
\end{algorithm}

\subsection{FedVRZO}

FedVRZO uses the same client trajectory and server update. Its local estimator
uses independent sample--direction pairs. At query point \(x_{i,t,k}\), client
\(i\) draws \(N_t\) independent sample--direction pairs:
\begin{equation}
\mathcal W_{i,t,k}
:=\{(\xi_{i,t,k}^{j},v_{i,t,k}^{j})\}_{j=1}^{N_t}.
\label{eq:fedvrzo_local_batch}
\end{equation}
Each sample--direction pair produces the elementary two-point estimate in
\eqref{eq:elementary_zo_estimator}. The client forms their sample mean:
\begin{equation}
    g_{i,t,k}:=\frac{1}{N_t}\sum_{j=1}^{N_t}
    \widehat g_{i,\eta_t}
    \bigl(x_{i,t,k};\xi_{i,t,k}^{j},v_{i,t,k}^{j}\bigr).
\label{eq:fedvrzo_local_direction}
\end{equation}
For each sample--direction pair, the positive and negative evaluations use the
same sample.
Conditional on the current history, the mean of \(g_{i,t,k}\) is
\(\nabla f_{i,\eta_t}(x_{i,t,k})\). The second moment of sampling error has a
\(1/N_t\)-scaled bound. Therefore, one local direction requires \(2N_t\)
sample-function evaluations.

\begin{algorithm}[t]
\caption{FedVRZO}
\label{alg:fedvrzo}
\begin{algorithmic}[1]
\REQUIRE \(x_0\in\X\), rounds \(T\), local evaluations \(K\), weights
\(\{p_i\}_{i=1}^{M}\), and sequences
\(\{\eta_t,\gamma_t,\alpha_t,N_t\}\)
\FOR{\(t=0,\ldots,T-1\)}
    \STATE Server broadcasts \(x_t,\eta_t,\gamma_t,N_t\).
    \FOR{each client \(i=1,\ldots,M\) in parallel}
        \STATE Set \(x_{i,t,0}=x_t\) and \(d_{i,t}=0\).
        \FOR{\(k=0,\ldots,K-1\)}
            \STATE Draw the sample--direction pairs \(\mathcal W_{i,t,k}\) from
            \eqref{eq:fedvrzo_local_batch}.
            \STATE Compute \(g_{i,t,k}\) from
            \eqref{eq:fedvrzo_local_direction} and set
            \(d_{i,t}\leftarrow d_{i,t}+g_{i,t,k}/K\).
            \IF{\(k<K-1\)}
                \STATE Set \(x_{i,t,k+1}\leftarrow
                \projX(x_{i,t,k}-\gamma_tg_{i,t,k})\).
            \ENDIF
        \ENDFOR
        \STATE Upload \(d_{i,t}\) to the server.
    \ENDFOR
    \STATE Server sets \(d_t\leftarrow\sum_{i=1}^{M}p_id_{i,t}\) and
    \(x_{t+1}\leftarrow\projX(x_t-\alpha_td_t)\).
\ENDFOR
    \STATE Draw \(R\) uniformly from \(\{0,\ldots,T-1\}\), independently of
    all algorithmic randomness.
\ENSURE \(x_R\)
\end{algorithmic}
\end{algorithm}

Both algorithms return the randomized server iterate \(x_R\). Independence
and uniform sampling of \(R\) imply that, for every nonnegative measurable
function \(\phi\),
\begin{equation}
    \E[\phi(x_R)]
    =\frac1T\sum_{t=0}^{T-1}\E[\phi(x_t)].
\label{eq:random_output_identity}
\end{equation}
Thus, bounds obtained by averaging the descent relation over the server
trajectory are guarantees for the stated algorithmic output.

%

\section{Convergence Analysis}
\label{sec:convergence}

This section establishes the convergence guarantees for FedZOO and FedVRZO.
We first introduce the common conditions and analytical tools used for both
methods. The sampling errors and convergence properties of FedZOO and
FedVRZO are then analyzed separately, followed by a comparison under a fixed
function-query budget.
\subsection{Common Conditions}

For \(t\geq0\) and \(k\in\{0,\ldots,K-1\}\), let
\(\mathcal F_{t,k}\) denote the information available immediately before the
random samples and directions at local step \(k\) are generated. In particular,
\(x_t\) and every \(x_{i,t,k}\) are \(\mathcal F_{t,k}\)-measurable. For
FedZOO, once the data batches at that step have been drawn, the enlarged
filtration is defined as
\begin{equation}
    \mathfrak H_{t,k}
    :=\mathcal F_{t,k}\vee
    \sigma\!\left(\{\mathcal B_{i,t,k}\}_{i=1}^{M}\right).
\label{eq:within_step_filtration}
\end{equation}
We use \(\mathcal F_{t,K}\) for the information available after the final
local direction evaluation. The FedZOO direction sets, the FedVRZO
sample--direction pairs, the resulting estimates, and each local projection
are included in
\(\mathcal F_{t,k+1}\). The history \(\mathcal F_{t+1,0}\) also contains the
server update at round \(t\).

The following two assumptions specify the regularity of the sample functions
and the conditional sampling structure that governs the analysis.

\begin{assumption}[Local Lipschitz continuity]
\label{ass:regularity}
\normalfont
For each client \(i\), the function \(F_i(\cdot,\xi_i)\) is Lipschitz
continuous on \(\X_{\eta_0}\) almost surely. Specifically, there exists a
nonnegative random variable \(\widetilde L_i(\xi_i)\) such that
\begin{equation}
    |F_i(x,\xi_i)-F_i(y,\xi_i)|
    \leq \widetilde L_i(\xi_i)\|x-y\|,
    \qquad x,y\in\X_{\eta_0},
\label{eq:sample_lipschitz}
\end{equation}
and
\begin{equation}
    L_i^2:=\E[\widetilde L_i(\xi_i)^2]<\infty.
\label{eq:client_lipschitz_moment}
\end{equation}
There exists \(x^\circ\in\X_{\eta_0}\) such that
\(\E[|F_i(x^\circ,\xi_i)|]<\infty\) for every \(i\). In addition, the global
objective is bounded below on \(\X_{\eta_0}\):
\begin{equation}
    f_{\inf}:=\inf_{x\in\X_{\eta_0}}f(x)>-\infty.
\label{eq:global_lower_bound}
\end{equation}
\end{assumption}

\begin{assumption}[Conditional sampling]
\label{ass:conditional_sampling}
\normalfont
At each local step \((t,k)\), the random collections generated by distinct
clients are conditionally independent given \(\mathcal F_{t,k}\). For
FedZOO, \(\mathcal B_{i,t,k}\) consists of \(B_t\) independent samples drawn
from \(\mathcal D_i\), and \(\mathcal V_{i,t,k}\) consists of \(Q_t\)
independent directions uniformly distributed on \(\mathbb S\). These two
collections are conditionally independent given \(\mathcal F_{t,k}\). For
FedVRZO, \(\mathcal W_{i,t,k}\) consists of \(N_t\) conditionally independent
sample--direction pairs \((\xi_{i,t,k}^{j},v_{i,t,k}^{j})\). Each sample is
drawn from \(\mathcal D_i\), each direction is uniform on \(\mathbb S\), and
the two components of each sample--direction pair are conditionally independent given
\(\mathcal F_{t,k}\).
\end{assumption}

Assumption~\ref{ass:regularity} allows the Lipschitz constant to depend on
both the client and the sampled data, without imposing differentiability on
\(F_i(\cdot,\xi_i)\). Random Lipschitz bounds and independent sample--direction pairs
are standard in nonsmooth stochastic zeroth-order analysis
\cite{marrinan2026zo}, while independent sampling of data and directions also
appears in federated zeroth-order analysis \cite{chen2023finegrained}.
Assumption~\ref{ass:conditional_sampling} ensures conditional centering of the
sampling errors and eliminates cross-client and cross-step covariance terms.

The aggregate and maximum Lipschitz constants, together with their
smoothing-dependent counterparts, are denoted by
\begin{equation}
\begin{gathered}
    \bar L:=\sum_{i=1}^{M}p_iL_i,\qquad
    L_{\max}:=\max_{1\leq i\leq M}L_i,\\
    L_{\eta}:=\frac{\sqrt n\,\bar L}{\eta},\qquad
    L_{\max,\eta}:=\frac{\sqrt n\,L_{\max}}{\eta}.
\end{gathered}
\label{eq:smoothing_constants}
\end{equation}
We quantify the client heterogeneity relevant to the convergence analysis
through the weighted variation among the smoothed local gradients:
\begin{equation}
    \zeta_\eta^2
    :=
    \sup_{x\in\X}
    \sum_{i=1}^{M}p_i
    \|\nabla f_{i,\eta}(x)-\nabla f_\eta(x)\|^2.
\label{eq:heterogeneity_measure}
\end{equation}
It vanishes when the clients have identical smoothed gradients, but may vary
with the smoothing radius.

\subsection{Auxiliary Results}

Three auxiliary results form the basis of both convergence analyses. They
establish regularity of the smoothed objectives, control the displacement
generated by projected local steps, and relate an inexact aggregate direction
to descent of the smoothed global objective.

\begin{lemma}[Smoothing regularity]
\label{lem:smoothing_regularities}
\normalfont
Suppose that the sample functions satisfy
Assumption~\ref{ass:regularity}. Then, for every \(0<\eta<\eta_0\),
\(i\in\{1,\ldots,M\}\), and \(x,y\in\X\),
\begin{subequations}
\begin{align}
\|\nabla f_{i,\eta}(x)\|&\leq L_i,
\label{eq:local_smoothed_gradient_bound}\\
\|\nabla f_{i,\eta}(x)-\nabla f_{i,\eta}(y)\|
&\leq\frac{\sqrt n\,L_i}{\eta}\|x-y\|,
\label{eq:local_smoothed_smoothness}\\
\|\nabla f_\eta(x)-\nabla f_\eta(y)\|
&\leq L_\eta\|x-y\|.
\label{eq:global_smoothed_smoothness}
\end{align}
\end{subequations}
The smoothed global objective satisfies \(f_\eta(x)\geq f_{\inf}\), and
\begin{equation}
0\leq\zeta_\eta^2\leq\sum_{i=1}^{M}p_iL_i^2<\infty.
\label{eq:heterogeneity_finite}
\end{equation}
\end{lemma}
\begin{proof}
	See Appendix~\ref{app:common_tools}. 
\end{proof}

The remaining auxiliary bounds involve the constants
\begin{equation}
\begin{gathered}
C_{\mathbb S}:=16\sqrt{2\pi},\qquad
C:=2\exp(2\tau_{\mathrm l}^2),\quad \tau_{\mathrm l}\in(0,1),\\
\rho_t:=\gamma_tK L_{\max,\eta},\qquad
\chi_K:=
\begin{cases}
0,&K=1,\\
1,&K\geq2,
\end{cases}
\end{gathered}
\label{eq:common_analysis_constants}
\end{equation}
where \(C_{\mathbb S}\) bounds the second moment of the spherical estimator,
\(C\) is determined by the threshold \(\tau_{\mathrm l}\), \(\rho_t\) is the
product of the cumulative local stepsize \(\gamma_tK\) and the largest local
smoothness constant, and \(\chi_K\) distinguishes single-step and multi-step
local updates. Equivalently, a prescribed \(\rho_t\) corresponds to
\(\gamma_t=\rho_t/(K L_{\max,\eta})\). The client-weighted Lipschitz moments are
denoted by
\begin{equation}
\mathfrak L_1:=\sum_{i=1}^{M}p_iL_i^2,\qquad
\mathfrak L_2:=\sum_{i=1}^{M}p_i^2L_i^2.
\label{eq:common_lipschitz_scales}
\end{equation}
If \(\mathfrak L_1=0\), every local objective is constant on
\(\X_{\eta_0}\), and the stationarity conclusions are immediate. We therefore
consider \(\mathfrak L_1>0\) and set
\begin{equation}
\varkappa:=\frac{\mathfrak L_2}{\mathfrak L_1}.
\label{eq:common_weight_ratio}
\end{equation}
Since \(p_i^2L_i^2\leq p_{\max}p_iL_i^2\), where
\(p_{\max}:=\max_i p_i\), we have
\(0<\varkappa\leq p_{\max}\leq1\).


To state the local-trajectory bound, fix a round \(t\) and a smoothing radius
\(\eta\). Let \(z_t\in\X\) be \(\mathcal F_{t,0}\)-measurable, and let each
adapted trajectory \(\{z_{i,t,k}\}_{k=0}^{K-1}\subseteq\X\) start from
\(z_{i,t,0}=z_t\), and let \(\omega_{i,t}^2\) be nonnegative and
\(\mathcal F_{t,0}\)-measurable. The projected recursion, its estimation
error, and the quantities entering the bound are
\begin{subequations}
\begin{align}
z_{i,t,k+1}
&=\projX(z_{i,t,k}-\gamma_t\widehat g_{i,t,k}),
\quad 0\leq k<K-1,
\label{eq:common_local_recursion}\\
e_{i,t,k}
&:=\widehat g_{i,t,k}-\nabla f_{i,\eta}(z_{i,t,k}),
\label{eq:common_estimation_error}\\
h_t&:=\nabla f_\eta(z_t),\qquad
\omega_t^2:=\sum_{i=1}^{M}p_i\omega_{i,t}^2,
\label{eq:common_direction_variance_definition}\\
\mathfrak d_t
&:=\frac{1}{K}\sum_{i=1}^{M}p_i\sum_{k=0}^{K-1}
[\nabla f_{i,\eta}(z_{i,t,k})-\nabla f_{i,\eta}(z_t)].
\label{eq:common_drift_definition}
\end{align}
\end{subequations}

\begin{lemma}[Projected local-trajectory bound]
\label{lem:common_local_trajectory}
\normalfont
Suppose that the preceding quantities are well defined and satisfy
\begin{align}
\E[e_{i,t,k}\mid\mathcal F_{t,k}]&=0,\nonumber\\
\E[\|e_{i,t,k}\|^2\mid\mathcal F_{t,k}]&\leq\omega_{i,t}^2.
\label{eq:common_estimator_conditions}
\end{align}
If \(\chi_K\rho_t\leq\tau_{\mathrm l}\), then
\begin{subequations}
\begin{align}
&\frac{1}{K}\sum_{i=1}^{M}p_i\sum_{k=0}^{K-1}
\E[\|z_{i,t,k}-z_t\|^2\mid\mathcal F_{t,0}]
\nonumber\\
&\quad\leq \chi_KC\gamma_t^2K^2
(\|h_t\|^2+\zeta_\eta^2+\omega_t^2),
\label{eq:common_trajectory_bound}\\
&\E[\|\mathfrak d_t\|^2\mid\mathcal F_{t,0}]
\leq \chi_KC\rho_t^2
(\|h_t\|^2+\zeta_\eta^2+\omega_t^2).
\label{eq:common_drift_bound}
\end{align}
\end{subequations}
\end{lemma}
\begin{proof}
	See Appendix~\ref{app:common_tools}.
\end{proof}
The server update is governed by the following deterministic descent
inequality.

\begin{lemma}[Inexact projected descent]
\label{lem:inexact_projected_descent}
\normalfont
Let \(x\in\X\), \(\widehat d\in\mathbb R^n\), and
\(x^+=\projX(x-a\widehat d)\), where \(0<a\leq1/L_\eta\). Then
\begin{equation}
f_\eta(x^+)\leq f_\eta(x)
-\frac{a}{2}\|\mathcal G_{\eta,a}(x)\|^2
+\frac{a}{2}\|\widehat d-\nabla f_\eta(x)\|^2.
\label{eq:inexact_projected_descent}
\end{equation}
\end{lemma}
\begin{proof}
See Appendix~\ref{app:common_tools}.
\end{proof}
\begin{remark}[Descent interpretation]
\label{rem:inexact_descent}
\normalfont
The inequality separates the decrease associated with projected stationarity
from the error in the update direction. When
\(\widehat d=\nabla f_\eta(x)\), the error term vanishes and each nonzero
projected gradient mapping produces a decrease in \(f_\eta\). 
\end{remark}
For a fixed smoothing radius \(\eta\), the initial descent budget used
in both analyses is denoted by
\begin{equation}
    \Delta_\eta:=f_\eta(x_0)-f_{\inf}.
\label{eq:common_initial_gap}
\end{equation}
For \(t=0,\ldots,T-1\), all convergence results below employ the common
parameter regime
\begin{equation}
\begin{aligned}
&\eta_t=\eta,\quad 0<\eta<\eta_0,\qquad
\alpha_t=\alpha,\quad 0<\alpha\leq L_\eta^{-1},\\
&\chi_K\rho_t\leq\tau_{\mathrm l}.
\end{aligned}
\label{eq:common_parameter_regime}
\end{equation}
The last condition is automatic for \(K=1\).
\subsection{Convergence Analysis of FedZOO}

 The batch-conditioned direction is defined as
\begin{equation}
\bar g_{i,t,k}:=\E\!\left[g_{i,t,k}\mid
\mathfrak H_{t,k}\right],
\label{eq:batch_conditioned_direction}
\end{equation}
where the conditional expectation is taken over the direction set drawn at the
current local step. The sampling errors of data and direction in client \(i\) are denoted as
\begin{equation}
\begin{aligned}
\varepsilon_{i,t,k}
&:=\bar g_{i,t,k}-\nabla f_{i,\eta}(x_{i,t,k}),\\
\nu_{i,t,k}&:=g_{i,t,k}-\bar g_{i,t,k}.
\end{aligned}
\label{eq:local_fedzoo_sampling_errors}
\end{equation}
The corresponding client-weighted averages over the \(K\) local query points are given by
\begin{subequations}
\begin{align}
\varepsilon_t
&:=\frac{1}{K}\sum_{i=1}^{M}p_i\sum_{k=0}^{K-1}\varepsilon_{i,t,k},
\label{eq:data_error_definition}\\
\nu_t
&:=\frac{1}{K}\sum_{i=1}^{M}p_i\sum_{k=0}^{K-1}\nu_{i,t,k},
\label{eq:direction_error_definition}\\
\widetilde\varepsilon_t&:=\varepsilon_t+\nu_t.
\label{eq:fedzoo_total_sampling_error}
\end{align}
\end{subequations}
For the algorithmic trajectories, the drift term in
Lemma~\ref{lem:common_local_trajectory} and the smoothed global direction
specialize to
\begin{subequations}
\begin{align}
\mathfrak d_t
&:=\frac{1}{K}\sum_{i=1}^{M}p_i\sum_{k=0}^{K-1}
\bigl[\nabla f_{i,\eta}(x_{i,t,k})
-\nabla f_{i,\eta}(x_t)\bigr],
\label{eq:common_drift_error_definition}\\
h_t&:=\nabla f_{\eta}(x_t).
\label{eq:common_smoothed_global_direction}
\end{align}
\end{subequations}
Consequently, the aggregated FedZOO direction admits the decomposition
\begin{equation}
	d_t=h_t+\mathfrak d_t+\widetilde\varepsilon_t.
	\label{eq:common_direction_decomposition}
\end{equation}
The drift term vanishes for \(K=1\). 
The single-step FedZOO sampling error satisfies the conditional
second-moment bound
\begin{equation}
	\E[\|\varepsilon_{i,t,k}+\nu_{i,t,k}\|^2
	\mid\mathcal F_{t,k}]
	\leq L_i^2
	\left(\frac{1}{B_t}
	+\frac{C_{\mathbb S}n}{Q_t}\right).
	\label{eq:fedzoo_single_step_error_bound}
\end{equation}
The local-trajectory analysis weights these bounds by \(p_i\), whereas
the squared aggregate error involves \(p_i^2\). Accordingly, the local
and aggregate sampling-error scales are defined as
\begin{subequations}
	\begin{align}
		\sigma_{\mathrm l,t}^{2}
		&:=\mathfrak L_1\left(\frac{1}{B_t}
		+\frac{C_{\mathbb S}n}{Q_t}\right),
        \label{eq:fedzoo_local_variance_scale}\\
		\sigma_{\mathrm g,t}^{2}
		&:=\mathfrak L_2\left(\frac{1}{B_t}
		+\frac{C_{\mathbb S}n}{Q_t}\right).
        \label{eq:fedzoo_aggregate_variance_scale}
	\end{align}
\end{subequations}
The scale \(\sigma_{\mathrm l,t}^2\) controls the client-weighted error in the
local-trajectory analysis, whereas conditional orthogonality gives
\(\E[\|\widetilde\varepsilon_t\|^2\mid\mathcal F_{t,0}]
\leq\sigma_{\mathrm g,t}^2/K\).

\begin{lemma}[FedZOO sampling error]
\label{lem:fedzoo_sampling_error}
\normalfont
Let Assumptions~\ref{ass:regularity} and
\ref{ass:conditional_sampling} hold. Then
\begin{subequations}
\label{eq:fedzoo_sampling_error_bounds}
\begin{align}
\E[\varepsilon_{i,t,k}\mid\mathcal F_{t,k}]&=0,
\label{eq:fedzoo_data_error_zero_mean}\\
\E[\nu_{i,t,k}\mid\mathfrak H_{t,k}]&=0,
\label{eq:fedzoo_direction_error_zero_mean}\\
\E[\|\varepsilon_{i,t,k}\|^2
\mid\mathcal F_{t,k}]&\leq\frac{L_i^2}{B_t},
\label{eq:fedzoo_data_error_second_moment}\\
\E[\|\nu_{i,t,k}\|^2
\mid\mathcal F_{t,k}]&\leq\frac{C_{\mathbb S}nL_i^2}{Q_t},
\label{eq:fedzoo_direction_error_second_moment}\\
\E[\|\widetilde\varepsilon_t\|^2
\mid\mathcal F_{t,0}]
&\leq\frac{\sigma_{\mathrm g,t}^{2}}{K}.
\label{eq:fedzoo_aggregate_error_bound}
\end{align}
\end{subequations}
\end{lemma}
\begin{proof}
	See Appendix~\ref{app:fedzoo_proofs}.	
\end{proof}	
Lemma~\ref{lem:fedzoo_sampling_error} controls the stochastic part of the
aggregated direction, but \(d_t\) is evaluated along the client-specific
trajectories rather than at \(x_t\). The remaining discrepancy is the drift
term \(\mathfrak d_t\). These two effects are combined in
the per-round residual
\begin{equation}
\mathcal E_t:=
\chi_KC\rho_t^2
(\bar L^2+\zeta_\eta^2+\sigma_{\mathrm l,t}^{2})
+\frac{\sigma_{\mathrm g,t}^{2}}{K}.
\label{eq:common_round_residual}
\end{equation}
The first term in \(\mathcal E_t\) accounts for the gradient discrepancy
induced by local trajectory displacement, while the second accounts for the
aggregated sampling error. 

\begin{lemma}[Direction-error bound of FedZOO]
\label{lem:fedzoo_direction_accuracy}
Suppose the conditions of Lemma~\ref{lem:fedzoo_sampling_error} hold and
\(\rho_t\leq\tau_{\mathrm l}\) when \(K\geq2\). Then
\begin{subequations}
\begin{align}
\E[\|\mathfrak d_t\|^2\mid\mathcal F_{t,0}]
&\leq\chi_KC\rho_t^2
(\bar L^2+\zeta_\eta^2+\sigma_{\mathrm l,t}^{2}),
\label{eq:fedzoo_drift_bound}\\
\E[\|d_t-h_t\|^2
\mid\mathcal F_{t,0}]&\leq2\mathcal E_t.
\label{eq:fedzoo_total_direction_bound}
\end{align}
\end{subequations}
\end{lemma}
\begin{proof}
See Appendix~\ref{app:fedzoo_proofs}.	
\end{proof}	
The bound in \eqref{eq:fedzoo_drift_bound} controls the
gradient discrepancy induced by the displacement of the local query
trajectories from the server iterate. The second inequality combines this drift
bound with \eqref{eq:fedzoo_aggregate_error_bound} and provides the direction-error term
required by the projected-descent inequality in
Lemma~\ref{lem:inexact_projected_descent}.
The server update is an inexact projected step for the smoothed global
objective, with direction error \(d_t-h_t\). Substituting
\eqref{eq:fedzoo_total_direction_bound} into
Lemma~\ref{lem:inexact_projected_descent} controls this error by
\(\mathcal E_t\). Thus, summing the resulting one-round descent relation gives the
following stationarity bound.

\begin{theorem}[Stationarity bound for FedZOO]
\label{thm:fedzoo_convergence}
\normalfont
Suppose that Assumptions~\ref{ass:regularity} and
\ref{ass:conditional_sampling} hold and that the parameters satisfy
\eqref{eq:common_parameter_regime}. Then
\begin{align}
\E\|\mathcal G_{\eta,\alpha}(x_R)\|^2
&=\frac1T\sum_{t=0}^{T-1}
\E\|\mathcal G_{\eta,\alpha}(x_t)\|^2\nonumber\\
&\leq\frac{2\Delta_\eta}{\alpha T}
+\frac2T\sum_{t=0}^{T-1}\mathcal E_t,
\label{eq:fedzoo_convergence_bound}
\end{align}
\end{theorem}
\begin{proof}
	See Appendix~\ref{app:fedzoo_proofs}.	
\end{proof}	
\begin{remark}
\label{rem:fedzoo_finite_time}
\normalfont
The first term in \eqref{eq:fedzoo_convergence_bound} is the initial descent
budget. By \eqref{eq:common_round_residual}, the second term contains the
local-drift and aggregate sampling errors. Their sampling factor
\(1/B_t+C_{\mathbb S}n/Q_t\) separates the contributions of data samples and
random directions, and the drift component vanishes when \(K=1\).
\end{remark}

\begin{corollary}
\label{cor:fedzoo_asymptotic}
\normalfont
If the conditions of Theorem~\ref{thm:fedzoo_convergence} hold and
\(T^{-1}\sum_{t=0}^{T-1}\mathcal E_t\to0\). Then
\(\E\|\mathcal G_{\eta,\alpha}(x_R)\|^2\to0\).
\end{corollary}

\begin{remark}
\label{rem:fedzoo_asymptotic_schedule}
The residual condition in Corollary~\ref{cor:fedzoo_asymptotic} follows if
\(T^{-1}\sum_{t=0}^{T-1}\rho_t^2\to0\) and
\(T^{-1}\sum_{t=0}^{T-1}(B_t^{-1}+Q_t^{-1})\to0\).
The condition on \(\rho_t\) is unnecessary when \(K=1\). For \(K\geq2\),
\begin{equation}
\rho_t\leq\frac{\rho_0}{\sqrt{t+1}},\quad
B_t\geq B_0(t+1),\quad Q_t\geq Q_0(t+1)
\label{eq:fedzoo_explicit_schedule}
\end{equation}
with \(0<\rho_0\leq\tau_{\mathrm l}\) and \(B_0,Q_0>0\) implies
\begin{equation}
\E\|\mathcal G_{\eta,\alpha}(x_R)\|^2
=\mathcal O\!\left(\frac{\log T}{T}\right).
\label{eq:fedzoo_log_rate}
\end{equation}
\end{remark}

For a prescribed total of \(T\) communication rounds, the round-independent
sampling budgets and local-step scale are selected as
\begin{subequations}
\begin{align}
\mathfrak s_T
&:=\max\left\{1,\frac{\varkappa T}{K}\right\},
\label{eq:common_horizon_scale}\\
B_T&:=\lceil \beta_B\mathfrak s_T\rceil,\qquad
Q_T:=\lceil \beta_Qn\mathfrak s_T\rceil,
\label{eq:fedzoo_horizon_schedule}\\
\rho_t&:=\frac{\rho_0}{\sqrt T},\qquad
0<\rho_0\leq\tau_{\mathrm l}.
\label{eq:fedzoo_horizon_rho_schedule}
\end{align}
\end{subequations}
where \(\beta_B,\beta_Q>0\) are constants. The budgets are kept fixed across
rounds, so that \(B_t=B_T\) and \(Q_t=Q_T\) for
\(t=0,\ldots,T-1\). The
sampling factor \(\mathfrak S_{B,Q}\) and the residual constant
\(\mathfrak R_{B,Q}\)
appearing in the convergence bound are given by
\begin{align}
\mathfrak S_{B,Q}&:=\frac1{\beta_B}+\frac{C_{\mathbb S}}{\beta_Q},\nonumber\\
\mathfrak R_{B,Q}&:=\chi_KC\rho_0^2
(\bar L^2+\zeta_\eta^2+\mathfrak L_1\mathfrak S_{B,Q})
+\mathfrak L_1\mathfrak S_{B,Q}.
\label{eq:fedzoo_residual_constant}
\end{align}

\begin{corollary}[Convergence rate of FedZOO]
\label{cor:fedzoo_horizon_rate}
\normalfont
Under the conditions of Theorem~\ref{thm:fedzoo_convergence}, set
\(\alpha=L_\eta^{-1}\) and adopt the choices in
\eqref{eq:fedzoo_horizon_schedule} and
\eqref{eq:fedzoo_horizon_rho_schedule}. Then
\begin{equation}
\E\|\mathcal G_{\eta,1/L_\eta}(x_R)\|^2
\leq\frac{2L_\eta\Delta_\eta+2\mathfrak R_{B,Q}}{T}.
\label{eq:fedzoo_horizon_rate}
\end{equation}
\end{corollary}
\begin{proof}
	See Appendix~\ref{app:fedzoo_proofs}.	
\end{proof}	

\begin{corollary}[Complexity of FedZOO]
\label{cor:fedzoo_complexity}
\normalfont
Under the parameter choices and conditions of
Corollary~\ref{cor:fedzoo_horizon_rate},
\begin{equation}
\E\|\mathcal G_{\eta,1/L_\eta}(x_R)\|
\leq
\sqrt{\frac{2L_\eta\Delta_\eta+2\mathfrak R_{B,Q}}{T}}.
\label{eq:fedzoo_average_mapping_norm}
\end{equation}
Thus, the expected mapping norm at the randomized output is no larger than \(\epsilon\) when
\(T\geq(2L_\eta\Delta_\eta+2\mathfrak R_{B,Q})\epsilon^{-2}\). The corresponding
communication requirements are \(\mathcal O(M\epsilon^{-2})\) uplink vectors
and \(\mathcal O(\epsilon^{-2})\) server broadcasts. The number of stochastic
function evaluations satisfies
\begin{subequations}
\begin{align}
\mathcal N_{\mathrm{FedZOO}}(T)
&=2MKT B_TQ_T
\nonumber\\
&=\mathcal O\!\Bigl(MKnT
+\frac{Mn\varkappa^2}{K}T^3\Bigr),
\label{eq:fedzoo_query_complexity_in_T}\\
\mathcal N_{\mathrm{FedZOO}}(\epsilon)
&=\mathcal O\!\left(MKn\epsilon^{-2}
+\frac{Mn\varkappa^2}{K}\epsilon^{-6}\right).
\label{eq:fedzoo_query_complexity_in_epsilon}
\end{align}
\end{subequations}
\end{corollary}
\vspace{2pt}
\begin{proof}
	See Appendix~\ref{app:fedzoo_proofs}.
\end{proof}

\subsection{Convergence Analysis of FedVRZO}

FedVRZO retains the direction-aggregation recursion of FedZOO but constructs
each local estimate from independent sample--direction pairs. Its local
sampling error and the corresponding client-weighted round average are
\begin{subequations}
\begin{align}
e_{i,t,k}
&:=g_{i,t,k}-\nabla f_{i,\eta}(x_{i,t,k}),
\label{eq:local_fedvrzo_sampling_error}\\
\widetilde e_t
&:=\frac{1}{K}\sum_{i=1}^{M}p_i\sum_{k=0}^{K-1}
e_{i,t,k}.
\label{eq:fedvrzo_total_sampling_error}
\end{align}
\end{subequations}
The common definitions of \(\mathfrak d_t\) and \(h_t\) in
\eqref{eq:common_drift_error_definition}--
\eqref{eq:common_smoothed_global_direction} apply to the FedVRZO trajectory;
thus, \eqref{eq:common_direction_decomposition} holds with
\(\widetilde e_t\) in place of \(\widetilde\varepsilon_t\). The corresponding
local and aggregate sampling-error scales are defined as
\begin{subequations}
\begin{align}
\widehat\sigma_{\mathrm l,t}^{2}
&:=\frac{C_{\mathbb S}n\mathfrak L_1}{N_t},
\label{eq:fedvrzo_local_variance_scale}\\
\widehat\sigma_{\mathrm g,t}^{2}
&:=\frac{C_{\mathbb S}n\mathfrak L_2}{N_t}.
\label{eq:fedvrzo_aggregate_variance_scale}
\end{align}
\end{subequations}
The next lemma formalizes these conditional mean and second-moment properties.
\begin{lemma}[FedVRZO sampling error]
\label{lem:fedvrzo_sampling_error}
\normalfont
Let Assumptions \ref{ass:regularity} and
\ref{ass:conditional_sampling} hold. Then
\begin{subequations}
\begin{align}
\E[e_{i,t,k}\mid\mathcal F_{t,k}]&=0,
\label{eq:fedvrzo_error_zero_mean}\\
\E[\|e_{i,t,k}\|^2\mid\mathcal F_{t,k}]
&\leq\frac{C_{\mathbb S}nL_i^2}{N_t},
\label{eq:fedvrzo_error_second_moment}\\
\E[\|\widetilde e_t\|^2
\mid\mathcal F_{t,0}]
&\leq\frac{\widehat\sigma_{\mathrm g,t}^{2}}{K}.
\label{eq:fedvrzo_aggregate_error_bound}
\end{align}
\end{subequations}
\end{lemma}
\begin{proof}
	See Appendix~\ref{app:fedvrzo_proofs}.	
\end{proof}	
The FedVRZO round residual is
\begin{equation}
\widehat{\mathcal E}_t
:=\chi_KC\rho_t^2
(\bar L^2+\zeta_\eta^2+\widehat\sigma_{\mathrm l,t}^2)
+\frac{\widehat\sigma_{\mathrm g,t}^2}{K}.
\label{eq:fedvrzo_round_residual}
\end{equation}

\begin{lemma}[Direction-error bound of FedVRZO]
\label{lem:fedvrzo_direction_accuracy}
\normalfont
Under Lemma \ref{lem:fedvrzo_sampling_error}, suppose
\(\rho_t\leq\tau_{\mathrm l}\) when \(K\geq2\). Then
\begin{subequations}
\begin{align}
\E[\|\mathfrak d_t\|^2\mid\mathcal F_{t,0}]
&\leq\chi_KC\rho_t^2
(\bar L^2+\zeta_\eta^2+\widehat\sigma_{\mathrm l,t}^{2}),
\label{eq:fedvrzo_drift_bound}\\
\E[\|d_t-h_t\|^2
\mid\mathcal F_{t,0}]&\leq2\widehat{\mathcal E}_t.
\label{eq:fedvrzo_total_direction_bound}
\end{align}
\end{subequations}
\end{lemma}
\begin{proof}
	See Appendix~\ref{app:fedvrzo_proofs}.	
\end{proof}	

\begin{theorem}[Stationarity bound for FedVRZO]
\label{thm:fedvrzo_convergence}
\normalfont
Suppose Assumptions~\ref{ass:regularity} and
\ref{ass:conditional_sampling} hold and that the parameters satisfy
\eqref{eq:common_parameter_regime}. Then
\begin{align}
\E\|\mathcal G_{\eta,\alpha}(x_R)\|^2
&=\frac1T\sum_{t=0}^{T-1}
\E\|\mathcal G_{\eta,\alpha}(x_t)\|^2\nonumber\\
&\leq\frac{2\Delta_\eta}{\alpha T}
+\frac2T\sum_{t=0}^{T-1}\widehat{\mathcal E}_t,
\label{eq:fedvrzo_convergence_bound}
\end{align}
\end{theorem}
\begin{proof}
	See Appendix~\ref{app:fedvrzo_proofs}.
\end{proof}

\begin{remark}
\label{rem:fedvrzo_finite_time}
\normalfont
By \eqref{eq:fedvrzo_round_residual}, independent sample--direction pairs
produce the single sampling factor \(n/N_t\). This factor enters both the
local-drift and aggregate sampling errors, while \(\rho_t^2\) controls the
displacement of the local query trajectories.
\end{remark}

\begin{corollary}
\label{cor:fedvrzo_asymptotic}
\normalfont
If the conditions of Theorem~\ref{thm:fedvrzo_convergence} hold and
\(T^{-1}\sum_{t=0}^{T-1}\widehat{\mathcal E}_t\to0\). Then
\(\E\|\mathcal G_{\eta,\alpha}(x_R)\|^2\to0\).
\end{corollary}

\begin{remark}
\label{rem:fedvrzo_asymptotic_schedule}
\normalfont
The residual condition in Corollary~\ref{cor:fedvrzo_asymptotic} follows if
\(T^{-1}\sum_{t=0}^{T-1}\rho_t^2\to0\) and
\(T^{-1}\sum_{t=0}^{T-1}N_t^{-1}\to0\).
The condition on \(\rho_t\) is unnecessary when \(K=1\). For \(K\geq2\),
\begin{equation}
\rho_t\leq\frac{\rho_0}{\sqrt{t+1}},\qquad N_t\geq N_0(t+1)
\label{eq:fedvrzo_explicit_schedule}
\end{equation}
with \(0<\rho_0\leq\tau_{\mathrm l}\) and \(N_0>0\) implies
\begin{equation}
\E\|\mathcal G_{\eta,\alpha}(x_R)\|^2
=\mathcal O\!\left(\frac{\log T}{T}\right).
\label{eq:fedvrzo_log_rate}
\end{equation}
\end{remark}

For a prescribed total of \(T\) communication rounds, FedVRZO employs the
common scaling factor \(\mathfrak s_T\) in
\eqref{eq:common_horizon_scale}. Its round-independent sampling budget and
local-step scale are
\begin{subequations}
\begin{align}
N_T&:=\lceil\beta_N n\mathfrak s_T\rceil,
\label{eq:fedvrzo_horizon_schedule}\\
\rho_t&:=\frac{\rho_0}{\sqrt T},\qquad
0<\rho_0\leq\tau_{\mathrm l}.
\label{eq:fedvrzo_horizon_rho_schedule}
\end{align}
\end{subequations}
where \(\beta_N>0\) is a constant. The budget is kept fixed across rounds,
so that \(N_t=N_T\) for \(t=0,\ldots,T-1\). The method-specific sampling factor
\(\mathfrak S_N\) and residual constant \(\mathfrak R_N\) are
given by
\begin{align}
\mathfrak S_N&:=\frac{C_{\mathbb S}}{\beta_N},\nonumber\\
\mathfrak R_N&:=\chi_KC\rho_0^2
(\bar L^2+\zeta_\eta^2+\mathfrak L_1\mathfrak S_N)
+\mathfrak L_1\mathfrak S_N.
\label{eq:fedvrzo_residual_constant}
\end{align}

\begin{corollary}[Convergence rate of FedVRZO]
\label{cor:fedvrzo_horizon_rate}
\normalfont
Under the conditions of Theorem~\ref{thm:fedvrzo_convergence}, set
\(\alpha=L_\eta^{-1}\) and adopt the choices in
\eqref{eq:fedvrzo_horizon_schedule} and
\eqref{eq:fedvrzo_horizon_rho_schedule}. Then
\begin{equation}
\E\|\mathcal G_{\eta,1/L_\eta}(x_R)\|^2
\leq\frac{2L_\eta\Delta_\eta+2\mathfrak R_N}{T}.
\label{eq:fedvrzo_horizon_rate}
\end{equation}
\end{corollary}
\begin{proof}
	See Appendix~\ref{app:fedvrzo_proofs}.
\end{proof}

\begin{corollary}[Complexity of FedVRZO]
\label{cor:fedvrzo_complexity}
\normalfont
Under the conditions and parameter choices of
Corollary~\ref{cor:fedvrzo_horizon_rate},
\begin{equation}
\E\|\mathcal G_{\eta,1/L_\eta}(x_R)\|
\leq
\sqrt{\frac{2L_\eta\Delta_\eta+2\mathfrak R_N}{T}}.
\label{eq:fedvrzo_average_mapping_norm}
\end{equation}
Thus, the expected mapping norm at the randomized output is no larger than \(\epsilon\) when
\(T\geq(2L_\eta\Delta_\eta+2\mathfrak R_N)\epsilon^{-2}\). The corresponding
communication requirements are \(\mathcal O(M\epsilon^{-2})\) uplink vectors
and \(\mathcal O(\epsilon^{-2})\) server broadcasts. The number of stochastic
function evaluations satisfies
\begin{subequations}
\begin{align}
\mathcal N_{\mathrm{FedVRZO}}(T)
&=2MKT N_T
=\mathcal O\!\left(MKnT+Mn\varkappa T^2\right),
\label{eq:fedvrzo_query_complexity_in_T}\\
\mathcal N_{\mathrm{FedVRZO}}(\epsilon)
&=\mathcal O\!\left(MKn\epsilon^{-2}
+Mn\varkappa\epsilon^{-4}\right).
\label{eq:fedvrzo_query_complexity_in_epsilon}
\end{align}
\end{subequations}
\end{corollary}
\begin{proof}
	See Appendix~\ref{app:fedvrzo_proofs}.
\end{proof}

The \(\mathcal O\)-notation in Corollaries~\ref{cor:fedzoo_complexity} and
\ref{cor:fedvrzo_complexity} treats the fixed smoothing radius and the
problem-dependent rate constants as constants. The displayed expressions
expose the sampling dependence on \(M\), \(K\), \(n\), \(\varkappa\), and
\(\epsilon\); the additional dimension dependence in \(L_\eta\) remains in
the rate constants. For uniform client weights, \(\varkappa=1/M\), and the
second terms in \eqref{eq:fedzoo_query_complexity_in_epsilon} and
\eqref{eq:fedvrzo_query_complexity_in_epsilon} reduce to
\(\mathcal O(n\epsilon^{-6}/(MK))\) and
\(\mathcal O(n\epsilon^{-4})\), respectively.

\begin{figure*}[!t]
	\centering
	\subfloat[Airplane]{
		\includegraphics[width=0.182\textwidth]
		{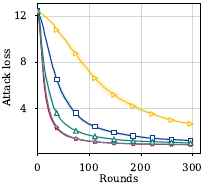}
	}
	\subfloat[Automobile]{
		\includegraphics[width=0.182\textwidth]
		{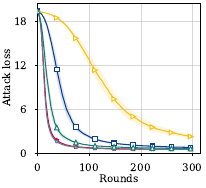}
	}
	\subfloat[Bird]{
		\includegraphics[width=0.182\textwidth]
		{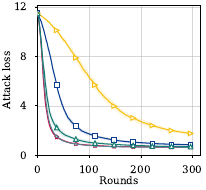}
	}
	\subfloat[Cat]{
		\includegraphics[width=0.182\textwidth]
		{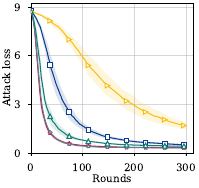}
	}
	\subfloat[Deer]{
		\includegraphics[width=0.182\textwidth]
		{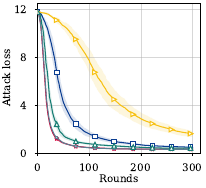}
	}\\[1mm]
	\subfloat[Dog]{
		\includegraphics[width=0.182\textwidth]
		{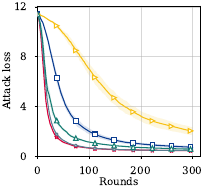}
	}
	\subfloat[Frog]{
		\includegraphics[width=0.182\textwidth]
		{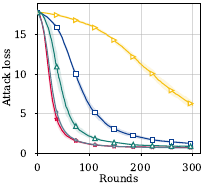}
	}
	\subfloat[Horse]{
		\includegraphics[width=0.182\textwidth]
		{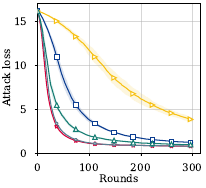}
	}
	\subfloat[Ship]{
		\includegraphics[width=0.182\textwidth]
		{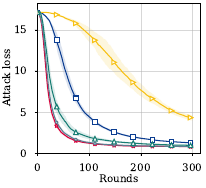}
	}
	\subfloat[Truck]{
		\includegraphics[width=0.182\textwidth]
		{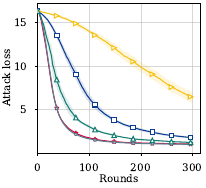}
	}

	\includegraphics[width=0.52\textwidth]{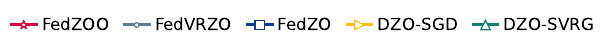}
	\vspace{-10pt}
	\caption{Test attack loss \(\ell_{\mathrm{atk}}\) versus communication rounds
		for the ten CIFAR-10 source classes. }\vspace{-10pt}
	\label{fig:cifar10_attack_loss_rounds}
\end{figure*}

\subsection{Sampling Efficiency}

Under a budget of \(P\) symmetric two-point differences per client and local query point,
FedZOO evaluates \(Q\) spherical directions using a
shared minibatch of \(B\) samples with \(BQ=P\), whereas FedVRZO draws \(N=P\)
independent sample--direction pairs. This budget corresponds to \(2P\) function
evaluations because each difference evaluates the sample function twice.
For \(P\geq C_{\mathbb S}n\),
\(\Psi_{\mathrm{FedZOO}}(P)\) denotes the smallest sampling-error factor in
\eqref{eq:fedzoo_single_step_error_bound} over the allocations satisfying
\(BQ=P\), while \(\Psi_{\mathrm{FedVRZO}}(P)\) denotes the corresponding
factor of FedVRZO with \(N=P\). These two factors are given by
\begin{equation}
\begin{aligned}
\Psi_{\mathrm{FedZOO}}(P)
&:=\inf_{\substack{B,Q>0\\BQ=P}}
\left(\frac1B+\frac{C_{\mathbb S}n}{Q}\right)
=2\sqrt{\frac{C_{\mathbb S}n}{P}},\\
\Psi_{\mathrm{FedVRZO}}(P)&:=\frac{C_{\mathbb S}n}{P}.
\end{aligned}
\label{eq:fixed_budget_sampling_factors}
\end{equation}

\begin{remark}[Interpretation of the sampling factors]
\normalfont
The continuous relaxation of the FedZOO allocation is
\(B^\star=\sqrt{P/(C_{\mathbb S}n)}\) and
\(Q^\star=\sqrt{C_{\mathbb S}nP}\). The condition
\(P\geq C_{\mathbb S}n\) ensures \(B^\star\geq1\). Integer budgets preserve
the same scaling after rounding, up to a constant-factor change in the query
budget. The ratio between the two factors is
\begin{equation}
\frac{\Psi_{\mathrm{FedVRZO}}(P)}{\Psi_{\mathrm{FedZOO}}(P)}
=\frac12\sqrt{\frac{C_{\mathbb S}n}{P}}.
\label{eq:fixed_budget_sampling_ratio}
\end{equation}
Thus, the optimized FedZOO factor decreases as \(P^{-1/2}\), whereas the
FedVRZO factor decreases as \(P^{-1}\). This comparison concerns only the
sampling component at a fixed local query point; the local-trajectory term
remains controlled by
\(\rho_t\), \(K\), and the smoothed-gradient heterogeneity \(\zeta_\eta^2\).
\end{remark}

At the round level, Lemmas~\ref{lem:fedzoo_sampling_error} and
\ref{lem:fedvrzo_sampling_error} give the respective aggregate sampling-error
bounds
\begin{equation}
\frac{\mathfrak L_2}{K}\Psi_{\mathrm{FedZOO}}(P)
\quad\text{and}\quad
\frac{\mathfrak L_2}{K}\Psi_{\mathrm{FedVRZO}}(P).
\label{eq:fixed_budget_aggregate_errors}
\end{equation}
Since \(\varkappa=\mathfrak L_2/\mathfrak L_1\), requiring each bound in
\eqref{eq:fixed_budget_aggregate_errors} to be no larger than
\(\mathfrak L_1/T\) and solving for \(P\) gives, in the regime
\(\varkappa T/K\geq1\),
\begin{equation}
P_{\mathrm{FedZOO}}\geq4C_{\mathbb S}n
\left(\frac{\varkappa T}{K}\right)^2,\quad
P_{\mathrm{FedVRZO}}\geq C_{\mathbb S}n\frac{\varkappa T}{K}.
\label{eq:fixed_budget_required_queries}
\end{equation}
Multiplying these budgets by \(2MKT\) recovers the \(T^3\) and \(T^2\)
sampling terms in \eqref{eq:fedzoo_query_complexity_in_T} and
\eqref{eq:fedvrzo_query_complexity_in_T}, respectively.

\section{Numerical Experiments}
\label{sec:experiments}
The experiments evaluate FedZOO and FedVRZO on black-box
attacks \cite{chen2017zoo,ilyas2018blackbox} against representative model-averaging ZO baselines. We report the attack
loss across communication rounds.

\subsection{Attack Problem and Experimental Protocol}

We evaluate the algorithms on CIFAR-10 using a pretrained deep neural network
with a clean test accuracy of \(83.51\%\). For each source class
\(c\in\{0,\ldots,9\}\), the corresponding experiment uses the training and
test images that are correctly classified as \(c\) before the attack and
optimizes one perturbation shared by all these images. For an image
\(z\in[-1/2,1/2]^n\) and perturbation \(\delta\), the attacked input is
\begin{equation}
 \mathsf T(z,\delta)
 :=\frac{1}{2}\tanh\!\left(\operatorname{arctanh}(2z)+\delta\right).
\label{eq:attack_transform}
\end{equation}
The transform is applied componentwise. At a boundary pixel,
\(\operatorname{arctanh}(\pm1)\) is interpreted by its limiting value, so the
corresponding transformed coordinate remains fixed at \(\pm1/2\).
Let \(s_j(\cdot)\) denote the logit of class \(j\). The attack component and
the sample objective are given by
\begin{equation}
\begin{aligned}
 \ell_{\mathrm{atk}}(\delta;z,c)
 &:={\left[s_c(\mathsf T(z,\delta))
 -\max_{j\neq c}s_j(\mathsf T(z,\delta))\right]}_+,\\
 \Phi_c(\delta;z)
 &:=\ell_{\mathrm{atk}}(\delta;z,c)
 +\lambda_{\mathrm{atk}}\|\mathsf T(z,\delta)-z\|^2.
\end{aligned}
\label{eq:attack_objective}
\end{equation}
Here, \([r]_+:=\max\{r,0\}\). For each client, \(z\) plays the role of the
random sample \(\xi_i\), and \(\Phi_c\) specifies the corresponding sample
function in \eqref{eq:intro_problem}. We set \(\lambda_{\mathrm{atk}}=1\), initialize \(\delta_0=0\), and constrain the
tanh-space perturbation by
\begin{equation}
 \X=\{\delta:\|\delta\|_{\infty}\leq1\}.
\label{eq:attack_constraint}
\end{equation}
The reported attack loss is the test average of
\(\ell_{\mathrm{atk}}\).

For each source class, the training images are randomly divided into five
nonoverlapping client datasets. All clients participate in every round, with
aggregation weights proportional to their local sample counts. All methods
use the same client split and zero initialization for a given seed, and the
curves report the mean over five seeds. Each run has \(300\) communication
rounds and uses \(\eta=10^{-2}\). FedZOO and FedVRZO set \(K=3\),
\(\gamma=10^{-3}\), and \(\alpha=10^{-2}\), with \((B,Q)=(32,5)\) and
\(N=160\), respectively.
The baselines are FedZO \cite{fedzo2022}, DZO-SGD, and DZO-SVRG. They use the
same attack objective, constraint, data partitions, and full-participation
protocol as the proposed methods. 

\subsection{Comparison with Model-Averaging Baselines}

Figure~\ref{fig:cifar10_attack_loss_rounds} compares the five methods under a
common communication horizon. At round \(50\), the average attack
loss is \(1.760\) for FedZOO and \(1.846\) for FedVRZO. The corresponding
values are \(3.381\), \(7.461\), and \(12.958\) for DZO-SVRG, FedZO, and
DZO-SGD. At round \(300\), FedVRZO and FedZOO reach average losses of
\(0.639\) and \(0.654\), respectively, compared with \(0.771\) for the
DZO-SVRG baseline. The sustained decrease in attack loss is consistent with
the finite-time stationarity guarantees established for the proposed
direction-aggregation recursions.

\section{Conclusion}
\label{sec:conclusion}

In this paper, we developed FedZOO and FedVRZO for constrained nonsmooth nonconvex stochastic optimization in federated settings.
FedZOO aggregates zeroth-order directions evaluated along projected local query trajectories and 
performs a single weighted projection at the server. FedVRZO preserves the same direction-aggregation 
recursion while constructing each local estimate from independent sample--direction pairs, 
thereby improving sampling efficiency. 
For both algorithms, we established nonasymptotic bounds on the expected squared projected gradient mapping 
and characterized their communication and function-query complexities. 
Numerical experiments on black-box adversarial attacks illustrated 
the empirical performance of the proposed methods.

\appendices
\section{Proofs of Auxiliary Results}
\label{app:common_tools}

\noindent\emph{Proof of Lemma~\ref{lem:smoothing_regularities}:}
We first establish the result for an arbitrary client \(i\). The integrability
condition at \(x^\circ\), together
with \eqref{eq:sample_lipschitz}, implies that \(F_i(x,\xi_i)\) is integrable
at every \(x\in\X_{\eta_0}\). Indeed,
\[
 |F_i(x,\xi_i)|
 \leq |F_i(x^\circ,\xi_i)|
 +\widetilde L_i(\xi_i)\|x-x^\circ\|,
\]
and \(\E[\widetilde L_i(\xi_i)]\leq L_i\) by Cauchy--Schwarz. Thus, the preceding
bound shows that \(\E[|F_i(x,\xi_i)|]<\infty\) for every
\(x\in\X_{\eta_0}\), so the local objective \(f_i\) is finite throughout
\(\X_{\eta_0}\). For
\(x,y\in\X_{\eta_0}\), Jensen's inequality and
\eqref{eq:sample_lipschitz} give
\begin{align}
 |f_i(x)-f_i(y)|
 &\leq \E[|F_i(x,\xi_i)-F_i(y,\xi_i)|]\nonumber\\
 &\leq \E[\widetilde L_i(\xi_i)]\|x-y\|
 \leq L_i\|x-y\|.
\label{eq:appendix_local_lipschitz}
\end{align}

Then, we transfer the Lipschitz estimate in
\eqref{eq:appendix_local_lipschitz} to the ball-smoothed function. For 
\(x,y\in\X\), the definition of \(f_{i,\eta}\) gives
\begin{align}
 |f_{i,\eta}(x)-f_{i,\eta}(y)|
 &\leq \E_u[|f_i(x+\eta u)-f_i(y+\eta u)|]\nonumber\\
 &\leq L_i\|x-y\|.
\label{eq:appendix_smoothed_lipschitz}
\end{align}
For an arbitrary \(x\in\X\), select any \(0<r<\eta_0-\eta\). Then
\(x+(\eta+r)\mathbb B\subseteq\X_{\eta_0}\), and \(f_i\) is
\(L_i\)-Lipschitz on this set by
\eqref{eq:appendix_local_lipschitz}. The standard uniform-ball smoothing
result \cite{marrinan2026zo} therefore implies that \(f_{i,\eta}\) is
continuously differentiable on \(x+r\mathbb B\) and that
\begin{equation}
 \nabla f_{i,\eta}(x)
 =\frac{n}{2\eta}\E_v
 \bigl[(f_i(x+\eta v)-f_i(x-\eta v))v\bigr].
\label{eq:appendix_smoothed_gradient_identity}
\end{equation}
Moreover, \eqref{eq:appendix_smoothed_lipschitz} holds on this neighborhood.
Consequently, the directional derivative along every unit vector \(z\)
satisfies
\begin{align*}
 |\langle\nabla f_{i,\eta}(x),z\rangle|
 &=\left|\lim_{s\to0}
 \frac{f_{i,\eta}(x+sz)-f_{i,\eta}(x)}{s}\right|\leq L_i.
\end{align*}
Taking the supremum over \(\|z\|=1\) gives
\eqref{eq:local_smoothed_gradient_bound}.

Using \eqref{eq:appendix_smoothed_gradient_identity}, for \(x,y\in\X\),
we define
\[
 \Theta_i(v):=f_i(x+\eta v)-f_i(x-\eta v)
 -f_i(y+\eta v)+f_i(y-\eta v).
\]
By \eqref{eq:appendix_local_lipschitz},
\begin{equation}
 |\Theta_i(v)|
 \leq2L_i\|x-y\|.
\label{eq:appendix_delta_bound}
\end{equation}
This bound implies that every unit vector \(z\) satisfies
\begin{align}
 &|\langle\nabla f_{i,\eta}(x)
 -\nabla f_{i,\eta}(y),z\rangle|\nonumber\\
 &\quad\leq\frac{nL_i}{\eta}\|x-y\|
 \E_v[|\langle v,z\rangle|]\nonumber\\
 &\quad\leq\frac{nL_i}{\eta}\|x-y\|
 \sqrt{\E_v[\langle v,z\rangle^2]}
 =\frac{\sqrt nL_i}{\eta}\|x-y\|.
\label{eq:appendix_local_smoothness_derivation}
\end{align}
The rotational symmetry gives
\(\E_v[\langle v,z\rangle^2]=1/n\). We take the supremum over unit \(z\)
to obtain \eqref{eq:local_smoothed_smoothness}. The weighted global gradient 
satisfies
\begin{align*}
 \|\nabla f_\eta(x)-\nabla f_\eta(y)\|
 &\leq\sum_{i=1}^{M}p_i
 \|\nabla f_{i,\eta}(x)-\nabla f_{i,\eta}(y)\|\\
 &\leq\frac{\sqrt n}{\eta}
 \left(\sum_{i=1}^{M}p_iL_i\right)\|x-y\|,
\end{align*}

For \(x\in\X\), \(u\in\mathbb B\), and \(0<\eta<\eta_0\),
\(x+\eta u\in\X_{\eta_0}\). Thus,
\begin{equation}
 f_\eta(x)=\E_u[f(x+\eta u)]\geq f_{\inf}.
\label{eq:appendix_smoothed_lower_bound}
\end{equation}
To bound the variation among the smoothed local gradients, we denote
\(\nabla f_{i,\eta}(x)\) by \(h_i\) and
\(\sum_i p_i h_i=\nabla f_\eta(x)\) by \(\bar h\). The weighted
square expansion gives
\begin{align}
 \sum_{i=1}^{M}p_i\|h_i-\bar h\|^2
 &=\sum_{i=1}^{M}p_i\|h_i\|^2-\|\bar h\|^2\nonumber\\
 &\leq\sum_{i=1}^{M}p_iL_i^2.
\label{eq:appendix_weighted_variance_identity}
\end{align}
Since the left-hand side is nonnegative, we take the supremum over
\(x\in\X\) to obtain \eqref{eq:heterogeneity_finite}.
\hfill\QED\par

\noindent\emph{Proof of Lemma~\ref{lem:common_local_trajectory}:}
If \(K=1\), then \(z_{i,t,0}=z_t\), so the two left-hand sides in
\eqref{eq:common_trajectory_bound} and \eqref{eq:common_drift_bound} are zero.
We now consider \(K\geq2\) at an arbitrary round \(t\) and denote the local
displacement and two gradient-related quantities by
\[
r_{i,k}:=z_{i,t,k}-z_t,\quad
h_{i,t}:=\nabla f_{i,\eta}(z_t),\quad
L_{i,\eta}:=\frac{\sqrt nL_i}{\eta}.
\]
For \(k=0,\ldots,K-1\), the conditional second moment of the displacement is
denoted by
\begin{equation}
\mathcal R_{i,k}:=\E[\|r_{i,k}\|^2\mid\mathcal F_{t,0}].
\label{eq:appendix_common_displacement_moment}
\end{equation}
Because \(z_{i,t,k}\in\X\), it is fixed by \(\projX\). Hence, for \(0\leq k<K-1\), projection
nonexpansivity gives, 
\begin{align}
 \|z_{i,t,k+1}-z_{i,t,k}\|
 &=\|\projX(z_{i,t,k}-\gamma_t\widehat g_{i,t,k})
 -\projX(z_{i,t,k})\|\nonumber\\
 &\leq\gamma_t\|\widehat g_{i,t,k}\|.
\label{eq:appendix_common_increment}
\end{align}
Since \(r_{i,0}=0\), we sum the increments and apply Cauchy--Schwarz to
obtain, for \(1\leq k\leq K-1\),
\begin{equation}
 \|r_{i,k}\|^2
 \leq\left(\sum_{j=0}^{k-1}
 \|z_{i,t,j+1}-z_{i,t,j}\|\right)^2
 \leq\gamma_t^2k\sum_{j=0}^{k-1}
 \|\widehat g_{i,t,j}\|^2.
\label{eq:appendix_common_local_path}
\end{equation}

The conditional zero-mean property in
\eqref{eq:common_estimator_conditions} eliminates the cross term in the
second moment of \(\widehat g_{i,t,j}\):
\begin{align}
\E[\|\widehat g_{i,t,j}\|^2\mid\mathcal F_{t,j}]
&=\|\nabla f_{i,\eta}(z_{i,t,j})\|^2
+\E[\|e_{i,t,j}\|^2\mid\mathcal F_{t,j}]\nonumber\\
&\leq\|\nabla f_{i,\eta}(z_{i,t,j})\|^2+\omega_{i,t}^2\nonumber\\
&\leq2\|h_{i,t}\|^2+\omega_{i,t}^2
+2L_{i,\eta}^2\|r_{i,j}\|^2.
\label{eq:appendix_common_direction_moment}
\end{align}
The last inequality uses \eqref{eq:local_smoothed_smoothness} and
\(\|a+b\|^2\leq2\|a\|^2+2\|b\|^2\).

Taking conditional expectation in \eqref{eq:appendix_common_local_path} and
using \eqref{eq:appendix_common_direction_moment} through the tower property
give the first inequality below. The second inequality follows from
\(k\leq K\):
\begin{align}
 \mathcal R_{i,k}
 &\leq\gamma_t^2k\sum_{j=0}^{k-1}
 \left(2\|h_{i,t}\|^2+\omega_{i,t}^2
 +2L_{i,\eta}^2\mathcal R_{i,j}\right)\nonumber\\
 &\leq A_{i,t}+c_{i,t}
 \sum_{j=0}^{k-1}\mathcal R_{i,j},
\label{eq:appendix_common_gronwall_recursion}
\end{align}
where
\[
A_{i,t}:=\gamma_t^2K^2(2\|h_{i,t}\|^2+\omega_{i,t}^2),
\qquad
c_{i,t}:=2\gamma_t^2K L_{i,\eta}^2.
\]
For the discrete Gronwall argument, we introduce
\(S_{i,k}:=A_{i,t}+c_{i,t}\sum_{j=0}^{k-1}\mathcal R_{i,j}\).
Then \(\mathcal R_{i,k}\leq S_{i,k}\), \(S_{i,1}=A_{i,t}\), and
\[
 S_{i,k+1}=S_{i,k}+c_{i,t}\mathcal R_{i,k}
 \leq(1+c_{i,t})S_{i,k}.
\]
The inequality follows from \(\mathcal R_{i,k}\leq S_{i,k}\) and
\(c_{i,t}\geq0\).
\begin{align}
 \mathcal R_{i,k}
 &\leq A_{i,t}(1+c_{i,t})^{k-1}\nonumber\\
 &\leq A_{i,t}\exp(c_{i,t}k)\nonumber\\
 &\leq e^{2\tau_{\mathrm l}^2}\gamma_t^2K^2
 (2\|h_{i,t}\|^2+\omega_{i,t}^2),
\label{eq:appendix_common_gronwall_bound}
\end{align}
The first inequality follows by induction from
\(S_{i,k+1}\leq(1+c_{i,t})S_{i,k}\), together with
\(\mathcal R_{i,k}\leq S_{i,k}\) and \(S_{i,1}=A_{i,t}\). Since
\(c_{i,t}\geq0\), the second inequality follows from
\(1+r\leq e^r\) for \(r\geq0\) and \(k-1\leq k\). The final inequality
follows from
\[
 c_{i,t}k\leq2\gamma_t^2K^2L_{i,\eta}^2
 \leq2\rho_t^2\leq2\tau_{\mathrm l}^2.
\]
The weighted variance identity at \(z_t\) gives
\begin{align}
 \sum_{i=1}^{M}p_i\|h_{i,t}\|^2
 &=\|h_t\|^2+
 \sum_{i=1}^{M}p_i\|h_{i,t}-h_t\|^2\nonumber\\
 &\leq\|h_t\|^2+\zeta_\eta^2.
\label{eq:appendix_common_client_gradient_moment}
\end{align}
We average \eqref{eq:appendix_common_gronwall_bound} over the clients and
local query points to obtain
\begin{align*}
 \frac1K\sum_{i=1}^{M}p_i\sum_{k=0}^{K-1}\mathcal R_{i,k}
 &\leq e^{2\tau_{\mathrm l}^2}\gamma_t^2K^2
 \left(2\sum_i p_i\|h_{i,t}\|^2+\omega_t^2\right)\\
 &\quad\leq2e^{2\tau_{\mathrm l}^2}\gamma_t^2K^2
 (\|h_t\|^2+\zeta_\eta^2+\omega_t^2),
\end{align*}
which is \eqref{eq:common_trajectory_bound} due to
\(C=2e^{2\tau_{\mathrm l}^2}\).

Then, we apply Jensen's inequality with weights \(p_i/K\) and invoke
\eqref{eq:local_smoothed_smoothness} to bound the drift term as
\begin{align}
 \|\mathfrak d_t\|^2
 &\leq\frac1K\sum_{i=1}^{M}p_i\sum_{k=0}^{K-1}
 \|\nabla f_{i,\eta}(z_{i,t,k})
 -\nabla f_{i,\eta}(z_t)\|^2\nonumber\\
 &\leq L_{\max,\eta}^2\frac1K
 \sum_{i=1}^{M}p_i\sum_{k=0}^{K-1}
 \|z_{i,t,k}-z_t\|^2.
\label{eq:appendix_common_drift_jensen}
\end{align}
We take conditional expectation in
\eqref{eq:appendix_common_drift_jensen} and apply
\eqref{eq:common_trajectory_bound}. The identity
\(\rho_t=\gamma_tK L_{\max,\eta}\) gives
\eqref{eq:common_drift_bound}.
\hfill\QED\par

\noindent\emph{Proof of Lemma~\ref{lem:inexact_projected_descent}:}
For notational convenience, we denote \(\nabla f_\eta(x)\) by \(h\), set
\(y=\projX(x-ah)\), and introduce
\[
\varphi_\ell(z):=\langle\ell,z-x\rangle+\frac{1}{2a}\|z-x\|^2.
\]
The definitions of \(x^+\) and \(y\) are equivalent to
\[
 x^+\in\arg\min_{z\in\X}\varphi_{\widehat d}(z),
 \qquad
 y\in\arg\min_{z\in\X}\varphi_h(z).
\]
We denote the direction error by \(e:=\widehat d-h\). Since
\(\varphi_{\widehat d}\) is
\(1/a\)-strongly convex and \(x^+\) is its constrained minimizer,
\begin{equation}
 \varphi_{\widehat d}(y)
 \geq\varphi_{\widehat d}(x^+)
 +\frac{1}{2a}\|y-x^+\|^2.
\label{eq:appendix_inexact_strong_convexity}
\end{equation}
The identity
\(\varphi_{\widehat d}(z)=\varphi_h(z)+\langle e,z-x\rangle\)
transforms \eqref{eq:appendix_inexact_strong_convexity} into
\begin{align}
 \varphi_h(x^+)-\varphi_h(y)
 &\leq-\frac{1}{2a}\|x^+-y\|^2
 -\langle e,x^+-y\rangle\nonumber\\
 &\leq\frac a2\|e\|^2,
\label{eq:appendix_inexact_model_comparison}
\end{align}
where the last inequality follows by completing the square.

The optimality condition for \(y=\projX(x-ah)\) is
\begin{equation}
 \left\langle h+\frac1a(y-x),z-y\right\rangle\geq0,
 \qquad z\in\X.
\label{eq:appendix_exact_projection_optimality}
\end{equation}
We set \(z=x\) in \eqref{eq:appendix_exact_projection_optimality} and obtain
\(\langle h,y-x\rangle\leq-a^{-1}\|y-x\|^2\), and hence
\begin{equation}
 \varphi_h(y)
 \leq-\frac{1}{2a}\|y-x\|^2
 =-\frac a2\|\mathcal G_{\eta,a}(x)\|^2.
\label{eq:appendix_exact_model_decrease}
\end{equation}
The \(L_\eta\)-smoothness of \(f_\eta\) and \(a\leq1/L_\eta\) imply
\begin{align}
 f_\eta(x^+)
 &\leq f_\eta(x)+\langle h,x^+-x\rangle
 +\frac{L_\eta}{2}\|x^+-x\|^2\nonumber\\
 &\leq f_\eta(x)+\varphi_h(x^+).
\label{eq:appendix_inexact_smoothness_step}
\end{align}
Combining \eqref{eq:appendix_inexact_model_comparison},
\eqref{eq:appendix_exact_model_decrease}, and
\eqref{eq:appendix_inexact_smoothness_step}, we can obtain the result.
\hfill\QED\par

\section{Proofs for FedZOO}
\label{app:fedzoo_proofs}

\noindent\emph{Proof of Lemma~\ref{lem:fedzoo_sampling_error}:}
For a single sample, we denote the gradient of the ball-smoothed sample
function by
\begin{equation}
\mathfrak G_{i,\eta}(x,\xi_i)
:=\frac{n}{2\eta}\E_v[
(F_i(x+\eta v,\xi_i)-F_i(x-\eta v,\xi_i))v].
\label{eq:appendix_fedzoo_sample_gradient}
\end{equation}
Fubini's theorem and \eqref{eq:smoothing_gradient_identity} give
\begin{equation}
 \E_{\xi_i}[\mathfrak G_{i,\eta}(x,\xi_i)]
 =\nabla f_{i,\eta}(x).
\label{eq:appendix_fedzoo_sample_unbiasedness}
\end{equation}
Because smoothing preserves the sample Lipschitz modulus, the argument used
for \eqref{eq:local_smoothed_gradient_bound}, now applied with \(\xi_i\)
fixed, gives
\begin{equation}
\|\mathfrak G_{i,\eta}(x,\xi_i)\|
\leq\widetilde L_i(\xi_i).
\label{eq:appendix_fedzoo_single_sample}
\end{equation}

We write the current data batch as
\(\mathcal B_{i,t,k}=\{\xi_{i,t,k}^{b}\}_{b=1}^{B_t}\). The conditional
expectation of \eqref{eq:fedzoo_local_direction} over the direction set gives
\begin{equation}
\bar g_{i,t,k}
=\frac1{B_t}\sum_{b=1}^{B_t}
\mathfrak G_{i,\eta}(x_{i,t,k},\xi_{i,t,k}^{b}).
\label{eq:appendix_fedzoo_batch_mean}
\end{equation}
For brevity, let \(\mathcal F=\mathcal F_{t,k}\), \(x=x_{i,t,k}\),
\(G_b:=\mathfrak G_{i,\eta}(x,\xi_{i,t,k}^{b})\),
\(\mu:=\nabla f_{i,\eta}(x)\), and \(W_b:=G_b-\mu\).
The iterate \(x\) is \(\mathcal F\)-measurable, and the conditional
distribution of each \(\xi_{i,t,k}^{b}\) is \(\mathcal D_i\). Therefore,
\eqref{eq:appendix_fedzoo_sample_unbiasedness} gives
\begin{align}
\E[G_b\mid\mathcal F]
&=\int \mathfrak G_{i,\eta}(x,\xi)\,\mathrm d\mathcal D_i(\xi)
=\mu,\quad \E[W_b\mid\mathcal F]=0,\nonumber\\
\E[\varepsilon_{i,t,k}\mid\mathcal F]
&=\frac1{B_t}\sum_{b=1}^{B_t}\E[W_b\mid\mathcal F]=0.
\label{eq:appendix_fedzoo_data_centered}
\end{align}

For \(b\neq c\), conditional independence gives
\(\E[\langle W_b,W_c\rangle\mid\mathcal F]
=\langle\E[W_b\mid\mathcal F],\E[W_c\mid\mathcal F]\rangle=0\).
Moreover, \eqref{eq:appendix_fedzoo_single_sample} and the conditional law of
\(\xi_{i,t,k}^{b}\) imply
\(\E[\|G_b\|^2\mid\mathcal F]
\leq\E[\widetilde L_i(\xi_{i,t,k}^{b})^2\mid\mathcal F]=L_i^2\).
Expanding the squared norm and using these zero cross moments gives
\begin{align}
\E[\|\varepsilon_{i,t,k}\|^2\mid\mathcal F]
&=\frac1{B_t^2}\sum_{b=1}^{B_t}
\E[\|W_b\|^2\mid\mathcal F]\nonumber\\
&=\frac1{B_t^2}\sum_{b=1}^{B_t}
\left(\E[\|G_b\|^2\mid\mathcal F]-\|\mu\|^2\right)\nonumber\\
&\leq\frac1{B_t^2}\sum_{b=1}^{B_t}
\E[\|G_b\|^2\mid\mathcal F]
\leq\frac{L_i^2}{B_t}.
\label{eq:appendix_fedzoo_data_variance}
\end{align}

We denote the sample-average function formed from \(\mathcal B_{i,t,k}\) by
\(\widehat F_{i,t,k}\) and define its direction estimate as
\begin{equation}
Y_{i,t,k}(v)=\frac{n}{2\eta}
[\widehat F_{i,t,k}(x_{i,t,k}+\eta v)
-\widehat F_{i,t,k}(x_{i,t,k}-\eta v)]v,
\label{eq:appendix_fedzoo_direction_variable}
\end{equation}
which implies that \(g_{i,t,k}\) is the average of \(Q_t\) independent copies of
\(Y_{i,t,k}\), conditioned on \(\mathfrak H_{t,k}\), and
\(\bar g_{i,t,k}\) is their conditional mean. The sample-average
function has Lipschitz modulus
\(\widehat L_{i,t,k}=B_t^{-1}\sum_b
\widetilde L_i(\xi_{i,t,k}^{b})\). The spherical
moment bound in \cite{marrinan2026zo} gives
\begin{equation}
\E[\|Y_{i,t,k}(v)\|^2\mid\mathfrak H_{t,k}]
\leq C_{\mathbb S}n\widehat L_{i,t,k}^2.
\label{eq:appendix_fedzoo_spherical_moment}
\end{equation}
As a result, the conditional mean and second moment of the direction error satisfy
\begin{subequations}
\begin{align}
 \E[\nu_{i,t,k}\mid\mathfrak H_{t,k}]&=0,
\label{eq:appendix_fedzoo_direction_centered}\\
 \E[\|\nu_{i,t,k}\|^2\mid\mathfrak H_{t,k}]
 &\leq\frac{C_{\mathbb S}n}{Q_t}\widehat L_{i,t,k}^2.
\label{eq:appendix_fedzoo_direction_variance_batch}
\end{align}
\end{subequations}
Jensen's inequality gives
\[
 \widehat L_{i,t,k}^2
 \leq\frac1{B_t}\sum_{b=1}^{B_t}
 \widetilde L_i(\xi_{i,t,k}^{b})^2.
\]
Next, we take conditional expectation in
\eqref{eq:appendix_fedzoo_direction_variance_batch} and obtain
\begin{equation}
 \E[\|\nu_{i,t,k}\|^2\mid\mathcal F_{t,k}]
 \leq\frac{C_{\mathbb S}nL_i^2}{Q_t},
\label{eq:appendix_fedzoo_direction_variance}
\end{equation}
which proves \eqref{eq:fedzoo_direction_error_zero_mean} and
\eqref{eq:fedzoo_direction_error_second_moment}.

To control the aggregate error, we denote the combined local sampling
increment by
\(e_{i,t,k}:=\varepsilon_{i,t,k}
+\nu_{i,t,k}\). Since
\(\varepsilon_{i,t,k}\) is
\(\mathfrak H_{t,k}\)-measurable, while
\eqref{eq:appendix_fedzoo_direction_centered} holds,
\begin{align}
 &\E[\langle\varepsilon_{i,t,k},
 \nu_{i,t,k}\rangle\mid\mathcal F_{t,k}]\nonumber\\
 &\quad=\E[\langle\varepsilon_{i,t,k},
 \E[\nu_{i,t,k}\mid\mathfrak H_{t,k}]\rangle
 \mid\mathcal F_{t,k}]=0.
\label{eq:appendix_fedzoo_mixed_orthogonality}
\end{align}
Combining \eqref{eq:appendix_fedzoo_data_variance},
\eqref{eq:appendix_fedzoo_direction_variance}, and
\eqref{eq:appendix_fedzoo_mixed_orthogonality}, we have
\begin{equation}
 \E[\|e_{i,t,k}\|^2\mid\mathcal F_{t,k}]
 \leq L_i^2\left(\frac1{B_t}
 +\frac{C_{\mathbb S}n}{Q_t}\right).
\label{eq:appendix_fedzoo_combined_local_variance}
\end{equation}

The variables \(e_{i,t,k}\) form conditionally centered sampling-error
increments. For two distinct clients at the same local step, conditional
independence given \(\mathcal F_{t,k}\) makes their cross moment zero. For
two different local steps \(k<\ell\), the earlier increment is
\(\mathcal F_{t,\ell}\)-measurable, whereas
\(\E[e_{j,t,\ell}\mid\mathcal F_{t,\ell}]=0\); hence the tower
property gives
\[
 \E[\langle e_{i,t,k},e_{j,t,\ell}\rangle
 \mid\mathcal F_{t,0}]=0.
\]
These orthogonality relations reduce the conditional second moment of
\(\widetilde\varepsilon_t\) to the sum of the individual second moments.
By \eqref{eq:appendix_fedzoo_combined_local_variance}, we have
\begin{align}
 \E[\|\widetilde\varepsilon_t\|^2
 \mid\mathcal F_{t,0}]
 &=\frac1{K^2}\sum_{i=1}^{M}p_i^2
 \sum_{k=0}^{K-1}
 \E[\|e_{i,t,k}\|^2\mid\mathcal F_{t,0}]\nonumber\\
 &\quad\leq\frac1{K^2}\sum_{i=1}^{M}p_i^2
 \sum_{k=0}^{K-1}L_i^2
 \left(\frac1{B_t}+\frac{C_{\mathbb S}n}{Q_t}\right)\nonumber\\
 &\quad=\frac{\sigma_{\mathrm g,t}^2}{K}.
\label{eq:appendix_fedzoo_aggregate_variance}
\end{align}
Thus, we can obtain the result in
\eqref{eq:fedzoo_aggregate_error_bound}.
\hfill\QED\par

\noindent\emph{Proof of Lemma~\ref{lem:fedzoo_direction_accuracy}:}
At each local step in FedZOO, the direction estimate admits the decomposition
\[
g_{i,t,k}
=\nabla f_{i,\eta}(x_{i,t,k})
+\varepsilon_{i,t,k}+\nu_{i,t,k}.
\]
Since \(\mathcal F_{t,k}\subseteq\mathfrak H_{t,k}\), the law of iterated
expectations applied to \eqref{eq:fedzoo_direction_error_zero_mean} gives
\(\E[\nu_{i,t,k}\mid\mathcal F_{t,k}]=0\). Combining this relation with
\eqref{eq:fedzoo_data_error_zero_mean}, we have
\[
 \E[g_{i,t,k}\mid\mathcal F_{t,k}]
 =\nabla f_{i,\eta}(x_{i,t,k}).
\]
Thus, the estimation error
\(e_{i,t,k}=\varepsilon_{i,t,k}+\nu_{i,t,k}\) has zero conditional mean and the conditional second moment
is bounded in~\eqref{eq:appendix_fedzoo_combined_local_variance}, so \eqref{eq:common_estimator_conditions} holds
with
\[
\omega_{i,t}^2=L_i^2\left(\frac1{B_t}
+\frac{C_{\mathbb S}n}{Q_t}\right),\qquad
\sum_i p_i\omega_{i,t}^2=\sigma_{\mathrm l,t}^2.
\]
We apply Lemma~\ref{lem:common_local_trajectory} with
\(z_t=x_t\), \(z_{i,t,k}=x_{i,t,k}\), and
\(\widehat g_{i,t,k}=g_{i,t,k}\) to obtain
\begin{align}
 \E[\|\mathfrak d_t\|^2\mid\mathcal F_{t,0}]
 &\leq C\rho_t^2
 (\|h_t\|^2+\zeta_\eta^2
 +\sigma_{\mathrm l,t}^2)
\label{eq:appendix_fedzoo_drift_prebound}
\end{align}
when \(K\geq2\). According to Lemma~\ref{lem:smoothing_regularities}, we have
\[
 \|h_t\|
 =\left\|\sum_i p_i\nabla f_{i,\eta}(x_t)\right\|
 \leq\sum_i p_iL_i=\bar L.
\]
By substituting this gradient bound into
\eqref{eq:appendix_fedzoo_drift_prebound}, we can obtain
\eqref{eq:fedzoo_drift_bound}. When \(K=1\), the same result follows from
\(\mathfrak d_t=0\) and the factor \(\chi_K=0\).
The decomposition in \eqref{eq:common_direction_decomposition} gives
\begin{align}
 \|d_t-h_t\|^2
 =\|\mathfrak d_t
 +\widetilde\varepsilon_t\|^2
 \leq2\|\mathfrak d_t\|^2
 +2\|\widetilde\varepsilon_t\|^2.
\label{eq:appendix_fedzoo_total_direction_split}
\end{align}
We take conditional expectation in
\eqref{eq:appendix_fedzoo_total_direction_split} and apply
\eqref{eq:fedzoo_drift_bound} and
\eqref{eq:fedzoo_aggregate_error_bound} to obtain
\begin{align*}
 \E[\|d_t-h_t\|^2
 \mid\mathcal F_{t,0}]
 &\leq2\chi_KC\rho_t^2
 (\bar L^2+\zeta_\eta^2+\sigma_{\mathrm l,t}^2)
 +\frac{2\sigma_{\mathrm g,t}^2}{K}\\
 &=2\mathcal E_t,
\end{align*}
which proves \eqref{eq:fedzoo_total_direction_bound}.
\hfill\QED\par

\noindent\emph{Proofs of Theorem~\ref{thm:fedzoo_convergence} and
Corollary~\ref{cor:fedzoo_asymptotic}:}
We apply Lemma~\ref{lem:inexact_projected_descent} with
\(x=x_t\), \(\widehat d=d_t\), and \(a=\alpha\) to obtain the one-step
projected-descent inequality.
Since \(x_t\) is \(\mathcal F_{t,0}\)-measurable, the
resulting one-step inequality can be conditioned on
\(\mathcal F_{t,0}\):
\begin{align}
 \E[f_\eta(x_{t+1})\mid\mathcal F_{t,0}]
 &\leq f_\eta(x_t)
 -\frac{\alpha}{2}
 \|\mathcal G_{\eta,\alpha}(x_t)\|^2\nonumber\\
 &\qquad+\frac{\alpha}{2}
 \E[\|d_t-h_t\|^2
 \mid\mathcal F_{t,0}]\nonumber\\
 &\leq f_\eta(x_t)
 -\frac{\alpha}{2}
 \|\mathcal G_{\eta,\alpha}(x_t)\|^2
 +\alpha\mathcal E_t,
\label{eq:appendix_fedzoo_conditional_descent}
\end{align}
where the last step uses \eqref{eq:fedzoo_total_direction_bound}. Then, we take
total expectation in \eqref{eq:appendix_fedzoo_conditional_descent} and sum
from \(t=0\) to \(T-1\) to obtain
\begin{align}
 \E[f_\eta(x_T)]
 &\leq f_\eta(x_0)
 -\frac{\alpha}{2}\sum_{t=0}^{T-1}
 \E\|\mathcal G_{\eta,\alpha}(x_t)\|^2
 +\alpha\sum_{t=0}^{T-1}\mathcal E_t.
\label{eq:appendix_fedzoo_telescoping}
\end{align}
According to Lemma~\ref{lem:smoothing_regularities}, we have \(f_\eta(x_T)\geq f_{\inf}\) almost surely. Rearranging
\eqref{eq:appendix_fedzoo_telescoping} gives
\begin{align*}
 \frac{\alpha}{2}\sum_{t=0}^{T-1}
 \E\|\mathcal G_{\eta,\alpha}(x_t)\|^2
 &\leq f_\eta(x_0)-f_{\inf}
 +\alpha\sum_{t=0}^{T-1}\mathcal E_t\\
 &=\Delta_\eta
 +\alpha\sum_{t=0}^{T-1}\mathcal E_t.
\end{align*}
Dividing by \(\alpha T/2\) and applying
\eqref{eq:random_output_identity} proves
\eqref{eq:fedzoo_convergence_bound}. Since \(\alpha>0\) and
\(\Delta_\eta<\infty\) are fixed, both terms on its right-hand side converge
to zero under the condition of Corollary~\ref{cor:fedzoo_asymptotic}; hence,
the corollary follows.
\hfill\QED\par

\noindent\emph{Verification of Remark~\ref{rem:fedzoo_asymptotic_schedule}:}
The conditions \(B_t,Q_t\geq1\) imply
\begin{equation}
 \sigma_{\mathrm l,t}^2
 \leq\mathfrak L_1(1+C_{\mathbb S}n)=:\mathfrak V_{B,Q}.
\label{eq:appendix_fedzoo_uniform_local_variance}
\end{equation}
According to \eqref{eq:fedzoo_aggregate_variance_scale} and the definition in \eqref{eq:common_round_residual},
we have
\begin{align}
 \frac1T\sum_{t=0}^{T-1}\mathcal E_t
 &\leq\chi_KC
 (\bar L^2+\zeta_\eta^2+\mathfrak V_{B,Q})
 \frac1T\sum_{t=0}^{T-1}\rho_t^2\nonumber\\
 &\quad+\frac{\mathfrak L_2}{K}
 \left[\frac1T\sum_{t=0}^{T-1}\frac1{B_t}
 +C_{\mathbb S}n\frac1T\sum_{t=0}^{T-1}\frac1{Q_t}\right].
\label{eq:appendix_fedzoo_average_residual}
\end{align}
Both terms on the right-hand side converge to zero under the sufficient
conditions stated in Remark~\ref{rem:fedzoo_asymptotic_schedule}.

Under \eqref{eq:fedzoo_explicit_schedule}, we introduce the constants
\[
 \mathfrak A_{B,Q}
 :=\mathfrak L_1\left(\frac1{B_0}
 +\frac{C_{\mathbb S}n}{Q_0}\right),\qquad
 \mathfrak B_{B,Q}
 :=\frac{\varkappa}{K}\mathfrak A_{B,Q}.
\]
For every \(t\geq0\),
\[
 \rho_t^2\leq\frac{\rho_0^2}{t+1},\qquad
 \sigma_{\mathrm l,t}^2
 \leq\frac{\mathfrak A_{B,Q}}{t+1},\qquad
 \frac{\sigma_{\mathrm g,t}^2}{K}
 \leq\frac{\mathfrak B_{B,Q}}{t+1},
\]
which imply
\begin{align}
 \mathcal E_t
 &\leq\frac{\chi_KC\rho_0^2}{t+1}
 \left(\bar L^2+\zeta_\eta^2
 +\frac{\mathfrak A_{B,Q}}{t+1}\right)
 +\frac{\mathfrak B_{B,Q}}{t+1}\nonumber\\
 &\leq\frac{\mathfrak C_{B,Q}}{t+1},
\label{eq:appendix_fedzoo_harmonic_residual}
\end{align}
where
\(\mathfrak C_{B,Q}:=\chi_KC\rho_0^2
(\bar L^2+\zeta_\eta^2+\mathfrak A_{B,Q})
+\mathfrak B_{B,Q}\). Finally, we substitute
\eqref{eq:appendix_fedzoo_harmonic_residual} into
\eqref{eq:fedzoo_convergence_bound} and use
\(\sum_{t=0}^{T-1}(t+1)^{-1}\leq1+\log T\) to obtain
\[
 \E\|\mathcal G_{\eta,\alpha}(x_R)\|^2
 \leq\frac{2\Delta_\eta}{\alpha T}
 +\frac{2\mathfrak C_{B,Q}(1+\log T)}{T},
\]
which proves \eqref{eq:fedzoo_log_rate}.
\hfill\QED\par

\noindent\emph{Proof of Corollary~\ref{cor:fedzoo_horizon_rate}:}
For brevity, we denote \(\mathfrak s_T\) by \(s\). Since
\(B_T=\lceil \beta_Bs\rceil\geq \beta_Bs\) and
\(Q_T=\lceil \beta_Qns\rceil\geq \beta_Qns\),
\begin{align}
 \frac1{B_T}+\frac{C_{\mathbb S}n}{Q_T}
 &\leq\frac1{\beta_Bs}+\frac{C_{\mathbb S}}{\beta_Qs}
 =\frac{\mathfrak S_{B,Q}}{s}.
\label{eq:appendix_fedzoo_horizon_sampling_factor}
\end{align}
These budget bounds imply
\begin{subequations}
\begin{align}
 \sigma_{\mathrm l,t}^2
 &\leq\frac{\mathfrak L_1\mathfrak S_{B,Q}}{s}
 \leq\mathfrak L_1\mathfrak S_{B,Q},
\label{eq:appendix_fedzoo_horizon_local_variance}\\
 \frac{\sigma_{\mathrm g,t}^2}{K}
 &\leq\frac{\mathfrak L_2\mathfrak S_{B,Q}}{Ks}
 =\frac{\varkappa\mathfrak L_1\mathfrak S_{B,Q}}{Ks}
 \leq\frac{\mathfrak L_1\mathfrak S_{B,Q}}{T}.
\label{eq:appendix_fedzoo_horizon_global_variance}
\end{align}
\end{subequations}
The last inequality uses
\(s=\max\{1,\varkappa T/K\}\geq\varkappa T/K\). Substituting
\(\rho_t^2=\rho_0^2/T\) and the bounds in
\eqref{eq:appendix_fedzoo_horizon_local_variance} and
\eqref{eq:appendix_fedzoo_horizon_global_variance} into
\eqref{eq:common_round_residual} gives
\begin{align}
 \mathcal E_t
 &\leq\frac{\chi_KC\rho_0^2}{T}
 (\bar L^2+\zeta_\eta^2
 +\mathfrak L_1\mathfrak S_{B,Q})
 +\frac{\mathfrak L_1\mathfrak S_{B,Q}}{T}\nonumber\\
 &=\frac{\mathfrak R_{B,Q}}{T}.
\label{eq:appendix_fedzoo_horizon_round_residual}
\end{align}
We substitute \eqref{eq:appendix_fedzoo_horizon_round_residual} into
\eqref{eq:fedzoo_convergence_bound} and set \(\alpha=1/L_\eta\) to obtain
\begin{align*}
 \E\|\mathcal G_{\eta,1/L_\eta}(x_R)\|^2
 &\leq\frac{2L_\eta\Delta_\eta}{T}
 +\frac{2\mathfrak R_{B,Q}}{T}\\
 &=\frac{2L_\eta\Delta_\eta+2\mathfrak R_{B,Q}}{T},
\end{align*}
which proves \eqref{eq:fedzoo_horizon_rate}.
\hfill\QED\par

\noindent\emph{Proof of Corollary~\ref{cor:fedzoo_complexity}:}
Jensen's inequality gives
\begin{align*}
 \E\|\mathcal G_{\eta,1/L_\eta}(x_R)\|
 &\leq\left(
 \E\|\mathcal G_{\eta,1/L_\eta}(x_R)\|^2\right)^{1/2}\\
 &\leq\sqrt{(2L_\eta\Delta_\eta+2\mathfrak R_{B,Q})/T}.
\end{align*}
The preceding inequality proves \eqref{eq:fedzoo_average_mapping_norm}.
Therefore,
\(T\geq(2L_\eta\Delta_\eta+2\mathfrak R_{B,Q})\epsilon^{-2}\) ensures that
the left-hand side of \eqref{eq:fedzoo_average_mapping_norm} is no larger
than \(\epsilon\). Since \(s\geq1\), the ceiling operation gives
\[
 B_T\leq(\beta_B+1)s,\qquad Q_T\leq(\beta_Q+1)ns.
\]
The resulting function-query count satisfies
\begin{align}
 \mathcal N_{\mathrm{FedZOO}}(T)
 &=2MKT B_TQ_T
 =\mathcal O(MKnT s^2)\nonumber\\
 &=\mathcal O\!\left(MKnT
 +\frac{Mn\varkappa^2}{K}T^3\right),
\label{eq:appendix_fedzoo_query_derivation}
\end{align}
where
\(s^2=\max\{1,(\varkappa T/K)^2\}
\leq1+(\varkappa T/K)^2\). Substituting
\(T=\mathcal O(\epsilon^{-2})\) into
\eqref{eq:appendix_fedzoo_query_derivation} gives
\eqref{eq:fedzoo_query_complexity_in_epsilon}. Thus, the two query bounds establish
\eqref{eq:fedzoo_query_complexity_in_T} and
\eqref{eq:fedzoo_query_complexity_in_epsilon}.
\hfill\QED\par

\section{Proofs for FedVRZO}
\label{app:fedvrzo_proofs}

\noindent\emph{Proof of Lemma~\ref{lem:fedvrzo_sampling_error}:}
For the \(j\)th sample--direction pair in \(\mathcal W_{i,t,k}\), we denote the associated
estimate by
\begin{equation}
Z_j:=\widehat g_{i,\eta}
(x_{i,t,k};\xi_{i,t,k}^{j},v_{i,t,k}^{j}).
\label{eq:appendix_fedvrzo_pair_variable}
\end{equation}
The iterate \(x_{i,t,k}\) is
\(\mathcal F_{t,k}\)-measurable. Conditional on \(\mathcal F_{t,k}\),
\(\xi_{i,t,k}^{j}\) is sampled from \(\mathcal D_i\),
\(v_{i,t,k}^{j}\) is sampled uniformly from \(\mathbb S\), and the two
variables are independent. Consequently,
\eqref{eq:elementary_zo_expectation} gives
\begin{equation}
 \E[Z_j\mid\mathcal F_{t,k}]
 =\nabla f_{i,\eta}(x_{i,t,k})
 =:\mu_{i,t,k}.
\label{eq:appendix_fedvrzo_pair_mean}
\end{equation}
For fixed \(\xi_{i,t,k}^{j}\), Assumption~\ref{ass:regularity} implies that
\(F_i(\cdot,\xi_{i,t,k}^{j})\) is
\(\widetilde L_i(\xi_{i,t,k}^{j})\)-Lipschitz continuous. The spherical
estimator second-moment bound in
\cite{marrinan2026zo} gives
\begin{equation}
\E[\|Z_j\|^2\mid\mathcal F_{t,k},\xi_{i,t,k}^{j}]
\leq C_{\mathbb S}n\widetilde L_i(\xi_{i,t,k}^{j})^2.
\label{eq:appendix_fedvrzo_pair_moment_conditional}
\end{equation}
We take expectation with respect to \(\xi_{i,t,k}^{j}\) and apply
\eqref{eq:client_lipschitz_moment} to obtain
\begin{equation}
 \E[\|Z_j\|^2\mid\mathcal F_{t,k}]
 \leq C_{\mathbb S}nL_i^2.
\label{eq:appendix_fedvrzo_pair_moment}
\end{equation}
The local estimation error can be written as
\[
 e_{i,t,k}
 =\frac1{N_t}\sum_{j=1}^{N_t}
 (Z_j-\mu_{i,t,k}),
\]
and conditional independence of the sample--direction pairs eliminates all cross terms when
the squared norm is expanded. Consequently,
\begin{subequations}
\begin{align}
 \E[e_{i,t,k}\mid\mathcal F_{t,k}]&=0,
\label{eq:appendix_fedvrzo_local_centered}\\
 \E[\|e_{i,t,k}\|^2
 \mid\mathcal F_{t,k}]
 &=\frac1{N_t^2}\sum_{j=1}^{N_t}
 \E[\|Z_j-\mu_{i,t,k}\|^2
 \mid\mathcal F_{t,k}]\nonumber\\
 &\leq\frac{C_{\mathbb S}nL_i^2}{N_t}.
\label{eq:appendix_fedvrzo_local_variance}
\end{align}
\end{subequations}

The aggregate error is a weighted sum of conditionally centered increments.
For different clients at the same local step, their conditional cross moments
are zero by Assumption~\ref{ass:conditional_sampling}. For \(k<\ell\),
\(e_{i,t,k}\) is
\(\mathcal F_{t,\ell}\)-measurable and
\(\E[e_{j,t,\ell}
\mid\mathcal F_{t,\ell}]=0\). Thus, by the tower property,
\[
 \E[\langle e_{i,t,k},
 e_{j,t,\ell}\rangle
 \mid\mathcal F_{t,0}]=0.
\]
The conditional orthogonality relations reduce the second moment in
\eqref{eq:fedvrzo_total_sampling_error} to a sum of individual second moments.
According to \eqref{eq:appendix_fedvrzo_local_variance},
\begin{align}
 \E[\|\widetilde e_t\|^2
 \mid\mathcal F_{t,0}]
 &=\frac1{K^2}\sum_{i=1}^{M}p_i^2
 \sum_{k=0}^{K-1}
 \E[\|e_{i,t,k}\|^2
 \mid\mathcal F_{t,0}]\nonumber\\
 &\leq\frac1{K^2}\sum_{i=1}^{M}p_i^2
 \sum_{k=0}^{K-1}\frac{C_{\mathbb S}nL_i^2}{N_t}
 =\frac{\widehat\sigma_{\mathrm g,t}^2}{K}.
\label{eq:appendix_fedvrzo_aggregate_variance}
\end{align}
Thus, we can obtain the result in
\eqref{eq:fedvrzo_aggregate_error_bound}.
\hfill\QED\par

\noindent\emph{Proof of Lemma~\ref{lem:fedvrzo_direction_accuracy}:}
The zero conditional mean in
\eqref{eq:appendix_fedvrzo_local_centered} and the conditional second-moment
bound in \eqref{eq:appendix_fedvrzo_local_variance} show that the local error
\(e_{i,t,k}\) satisfies \eqref{eq:common_estimator_conditions} with
\[
\omega_{i,t}^2=\frac{C_{\mathbb S}nL_i^2}{N_t},\qquad
\sum_i p_i\omega_{i,t}^2=\widehat\sigma_{\mathrm l,t}^2.
\]
We apply Lemma~\ref{lem:common_local_trajectory} with
\(z_t=x_t\), \(z_{i,t,k}=x_{i,t,k}\), and
\(\widehat g_{i,t,k}=g_{i,t,k}\). For \(K\geq2\), 
\begin{align}
 \E[\|\mathfrak d_t\|^2\mid\mathcal F_{t,0}]
 &\leq C\rho_t^2
 (\|h_t\|^2+\zeta_\eta^2
 +\widehat\sigma_{\mathrm l,t}^2).
\label{eq:appendix_fedvrzo_drift_prebound}
\end{align}
Using Lemma~\ref{lem:smoothing_regularities}, we have
\[
 \|h_t\|
 =\left\|\sum_i p_i\nabla f_{i,\eta}(x_t)\right\|
 \leq\sum_i p_iL_i=\bar L.
\]
We substitute this bound into \eqref{eq:appendix_fedvrzo_drift_prebound} to
obtain \eqref{eq:fedvrzo_drift_bound}. For \(K=1\), the drift
is zero and the same bound follows from \(\chi_K=0\).

For FedVRZO, the direction decomposition in
\eqref{eq:common_direction_decomposition} implies
\begin{align}
 \|d_t-h_t\|^2
 =\|\mathfrak d_t
 +\widetilde e_t\|^2
\leq2\|\mathfrak d_t\|^2
 +2\|\widetilde e_t\|^2.
\label{eq:appendix_fedvrzo_total_direction_split}
\end{align}
We take conditional expectation in
\eqref{eq:appendix_fedvrzo_total_direction_split} and apply
\eqref{eq:fedvrzo_drift_bound} and
\eqref{eq:fedvrzo_aggregate_error_bound} to obtain
\begin{align*}
 \E[\|d_t-h_t\|^2
 \mid\mathcal F_{t,0}]
 &\leq2\chi_KC\rho_t^2
 (\bar L^2+\zeta_\eta^2+\widehat\sigma_{\mathrm l,t}^2)
 +\frac{2\widehat\sigma_{\mathrm g,t}^2}{K}\\
 &=2\widehat{\mathcal E}_t,
\end{align*}
which proves \eqref{eq:fedvrzo_total_direction_bound}.
\hfill\QED\par

\noindent\emph{Proofs of Theorem~\ref{thm:fedvrzo_convergence} and
Corollary~\ref{cor:fedvrzo_asymptotic}:}
We apply Lemma~\ref{lem:inexact_projected_descent} with
\(x=x_t\), \(\widehat d=d_t\), and \(a=\alpha\) to obtain the one-step
projected-descent inequality. Then, we take conditional expectation with
respect to \(\mathcal F_{t,0}\) and apply
\eqref{eq:fedvrzo_total_direction_bound} to obtain
\begin{align}
 \E[f_\eta(x_{t+1})\mid\mathcal F_{t,0}]
 &\leq f_\eta(x_t)
 -\frac{\alpha}{2}
 \|\mathcal G_{\eta,\alpha}(x_t)\|^2
 +\alpha\widehat{\mathcal E}_t.
\label{eq:appendix_fedvrzo_conditional_descent}
\end{align}
Next, we take total expectation and sum from \(t=0\) to \(T-1\) to obtain
\begin{align}
 \E[f_\eta(x_T)]
 \leq f_\eta(x_0)
 -\frac{\alpha}{2}\sum_{t=0}^{T-1}
 \E\|\mathcal G_{\eta,\alpha}(x_t)\|^2
 +\alpha\sum_{t=0}^{T-1}\widehat{\mathcal E}_t.
\label{eq:appendix_fedvrzo_telescoping}
\end{align}
The lower bound \(f_\eta(x_T)\geq f_{\inf}\) and the definition in
\eqref{eq:common_initial_gap} transform
\eqref{eq:appendix_fedvrzo_telescoping} into
\[
 \frac{\alpha}{2}\sum_{t=0}^{T-1}
 \E\|\mathcal G_{\eta,\alpha}(x_t)\|^2
 \leq\Delta_\eta
 +\alpha\sum_{t=0}^{T-1}\widehat{\mathcal E}_t.
\]
Dividing by \(\alpha T/2\) and applying
\eqref{eq:random_output_identity} proves
\eqref{eq:fedvrzo_convergence_bound}. Since \(\alpha>0\) and
\(\Delta_\eta<\infty\) are fixed, both terms on its right-hand side converge
to zero under the condition of Corollary~\ref{cor:fedvrzo_asymptotic}; hence,
the corollary follows.
\hfill\QED\par

\noindent\emph{Verification of Remark~\ref{rem:fedvrzo_asymptotic_schedule}:}
The condition \(N_t\geq1\) implies
\begin{equation}
 \widehat\sigma_{\mathrm l,t}^2
 \leq C_{\mathbb S}n\mathfrak L_1=:\mathfrak V_N.
\label{eq:appendix_fedvrzo_uniform_local_variance}
\end{equation}
Equations \eqref{eq:fedvrzo_aggregate_variance_scale} and
\eqref{eq:fedvrzo_round_residual} give
\begin{align}
 \frac1T\sum_{t=0}^{T-1}\widehat{\mathcal E}_t
 &\leq\chi_KC
 (\bar L^2+\zeta_\eta^2+\mathfrak V_N)
 \frac1T\sum_{t=0}^{T-1}\rho_t^2\nonumber\\
 &\quad+\frac{C_{\mathbb S}n\mathfrak L_2}{K}
 \frac1T\sum_{t=0}^{T-1}\frac1{N_t}.
\label{eq:appendix_fedvrzo_average_residual}
\end{align}
The right-hand side converges to zero under the sufficient conditions stated
in Remark~\ref{rem:fedvrzo_asymptotic_schedule}.

Under \eqref{eq:fedvrzo_explicit_schedule}, we introduce the constants
\[
 \mathfrak A_N:=\frac{C_{\mathbb S}n\mathfrak L_1}{N_0},
 \qquad
 \mathfrak B_N:=\frac{C_{\mathbb S}n\mathfrak L_2}{KN_0}.
\]
For every \(t\geq0\),
\[
 \rho_t^2\leq\frac{\rho_0^2}{t+1},\qquad
 \widehat\sigma_{\mathrm l,t}^2
 \leq\frac{\mathfrak A_N}{t+1},\qquad
 \frac{\widehat\sigma_{\mathrm g,t}^2}{K}
 \leq\frac{\mathfrak B_N}{t+1}.
\]
These estimates imply
\begin{align}
 \widehat{\mathcal E}_t
 &\leq\frac{\chi_KC\rho_0^2}{t+1}
 \left(\bar L^2+\zeta_\eta^2
 +\frac{\mathfrak A_N}{t+1}\right)
 +\frac{\mathfrak B_N}{t+1}\nonumber\\
 &\leq\frac{\mathfrak C_N}{t+1},
\label{eq:appendix_fedvrzo_harmonic_residual}
\end{align}
where
\(\mathfrak C_N:=\chi_KC\rho_0^2
(\bar L^2+\zeta_\eta^2+\mathfrak A_N)
+\mathfrak B_N\). We substitute this bound into
\eqref{eq:fedvrzo_convergence_bound} to obtain
\[
 \E\|\mathcal G_{\eta,\alpha}(x_R)\|^2
 \leq\frac{2\Delta_\eta}{\alpha T}
 +\frac{2\mathfrak C_N(1+\log T)}{T},
\]
which proves \eqref{eq:fedvrzo_log_rate}.
\hfill\QED\par

\noindent\emph{Proof of Corollary~\ref{cor:fedvrzo_horizon_rate}:}
For brevity, we denote \(\mathfrak s_T\) by \(s\). Since
\(N_T=\lceil\beta_N ns\rceil\geq\beta_N ns\),
\begin{subequations}
\begin{align}
 \widehat\sigma_{\mathrm l,t}^2
 &=\frac{C_{\mathbb S}n\mathfrak L_1}{N_T}
 \leq\frac{\mathfrak L_1\mathfrak S_N}{s}
 \leq\mathfrak L_1\mathfrak S_N,
\label{eq:appendix_fedvrzo_horizon_local_variance}\\
 \frac{\widehat\sigma_{\mathrm g,t}^2}{K}
 &=\frac{C_{\mathbb S}n\mathfrak L_2}{KN_T}
 \leq\frac{\varkappa\mathfrak L_1\mathfrak S_N}{Ks}
 \leq\frac{\mathfrak L_1\mathfrak S_N}{T}.
\label{eq:appendix_fedvrzo_horizon_global_variance}
\end{align}
\end{subequations}
The last step uses \(s\geq\varkappa T/K\). We combine these bounds with
\(\rho_t^2=\rho_0^2/T\) to obtain
\begin{align}
 \widehat{\mathcal E}_t
 &\leq\frac{\chi_KC\rho_0^2}{T}
 (\bar L^2+\zeta_\eta^2
 +\mathfrak L_1\mathfrak S_N)
 +\frac{\mathfrak L_1\mathfrak S_N}{T}\nonumber\\
 &=\frac{\mathfrak R_N}{T}.
\label{eq:appendix_fedvrzo_horizon_round_residual}
\end{align}
Finally, we substitute \eqref{eq:appendix_fedvrzo_horizon_round_residual} into
\eqref{eq:fedvrzo_convergence_bound} and set \(\alpha=1/L_\eta\) to obtain
\begin{align*}
 \E\|\mathcal G_{\eta,1/L_\eta}(x_R)\|^2
 &\leq\frac{2}{T}\left(L_\eta\Delta_\eta+\mathfrak R_N\right),
\end{align*}
which proves \eqref{eq:fedvrzo_horizon_rate}.
\hfill\QED\par

\noindent\emph{Proof of Corollary~\ref{cor:fedvrzo_complexity}:}
Jensen's inequality gives
\begin{align*}
 \E\|\mathcal G_{\eta,1/L_\eta}(x_R)\|
 &\leq\left(
 \E\|\mathcal G_{\eta,1/L_\eta}(x_R)\|^2\right)^{1/2}\\
 &\leq\sqrt{(2L_\eta\Delta_\eta+2\mathfrak R_N)/T}.
\end{align*}
The preceding inequality proves \eqref{eq:fedvrzo_average_mapping_norm}.
Therefore,
\(T\geq(2L_\eta\Delta_\eta+2\mathfrak R_N)\epsilon^{-2}\) is sufficient.
The inequality \(s\geq1\) ensures that
\(N_T=\lceil\beta_N ns\rceil\leq(\beta_N+1)ns\). The resulting query count
satisfies
\begin{align}
 \mathcal N_{\mathrm{FedVRZO}}(T)
 &=2MKTN_T
 =\mathcal O(MKnT s)\nonumber\\
 &=\mathcal O\!\left(MKnT+Mn\varkappa T^2\right),
\label{eq:appendix_fedvrzo_query_derivation}
\end{align}
where
\(s=\max\{1,\varkappa T/K\}\leq1+\varkappa T/K\). Substituting
\(T=\mathcal O(\epsilon^{-2})\) into
\eqref{eq:appendix_fedvrzo_query_derivation} gives
\eqref{eq:fedvrzo_query_complexity_in_epsilon}. Thus, the two query bounds establish
\eqref{eq:fedvrzo_query_complexity_in_T} and
\eqref{eq:fedvrzo_query_complexity_in_epsilon}.
\hfill\QED\par

\bibliographystyle{IEEEtran}
\bibliography{mybib}

\end{document}